\documentclass [journal,onecolumn,11pt]{IEEEtran}
\usepackage{amsfonts,amsmath,amssymb}
\usepackage{indentfirst, setspace}
\usepackage{url,float}
\usepackage{color}
\usepackage[a4paper,ignoreall]{geometry}
\usepackage{longtable}
\usepackage{graphicx}
\usepackage[a4paper,ignoreall]{geometry}
\usepackage{multicol}
\usepackage{stfloats}
\usepackage{enumerate}
\usepackage{cite}
\usepackage{amsthm}
\usepackage{booktabs}
\usepackage{mathrsfs}
\usepackage{multirow}
\usepackage{amssymb}
\usepackage{verbatim}
\usepackage{cases}
\usepackage[square, comma, sort&compress, numbers]{natbib}
\usepackage[overload]{empheq}

\def\qu#1 {\fbox {\footnote {\ }}\ \footnotetext { From Qu: {\color{red}#1}}}
\def\hqu#1 {}
\def\kq#1 {\fbox {\footnote {\ }}\ \footnotetext { From KangQuan: {\color{blue}#1}}}
\def\hkq#1 {}

\usepackage{longtable}

\newtheorem{Th}{Theorem}[section]

\newtheorem{Prop}[Th]{Proposition}

\newtheorem{Lemma}[Th]{Lemma}
\newtheorem{Def}[Th]{Definition}

\newtheorem{Rem}[Th]{Remark}

\newcommand{\tr}{{\rm Tr}}

\newcommand{\gf}{{\mathbb F}}

\makeatletter
\newcommand{\figcaption}{\def\@captype{figure}\caption}
\newcommand{\tabcaption}{\def\@captype{table}\caption}
\makeatother

\begin{document}
	\title{Further study of $2$-to-$1$ mappings over $\gf_{2^n}$ }
	\author{{Kangquan Li,  Sihem Mesnager and Longjiang Qu }\\
	\space
This work is dedicated in memory of Vladimir I. Levenshtein
	\thanks{Kangquan Li and Longjiang Qu  are with the College of Liberal Arts and Sciences,
		National University of Defense Technology, Changsha, 410073, China.
		Longjiang Qu is also with the State Key Laboratory of Cryptology, Beijing, 100878, China. Sihem Mesnager is with the department of Mathematics, University of Paris VIII, 93526 Saint-Denis, France, University of Paris XIII, CNRS, LAGA UMR 7539, Sorbonne Paris Cit\'e, 93430 Villetaneuse,
France, and Telecom ParisTech 75013 Paris.
		E-mail: likangquan11@nudt.edu.cn, smesnager@univ-paris8.fr, ljqu\_happy@hotmail.com. 	The research of Longjiang Qu is supported by the Nature Science Foundation of China (NSFC) under Grant  61722213, 11531002, 61572026,   National Key R$\&$D Program of China (No.2017YFB0802000),  and the Open Foundation of State Key Laboratory of Cryptology.
 }}
	\maketitle{}
\begin{abstract}
	$2$-to-$1$ mappings over finite fields play an important role in symmetric cryptography, in particular in the constructions of APN functions, bent functions, semi-bent functions and so on. Very recently, Mesnager and Qu \cite{MQ2019} provided a systematic study of $2$-to-$1$ mappings over finite fields. In particular, they determined all $2$-to-$1$ mappings of degree at most 4 over any finite fields. In addition, another research direction is to consider $2$-to-$1$ polynomials with few terms.  Some results about $2$-to-$1$ monomials and binomials have been obtained in \cite{MQ2019}.
	
Motivated by their work, in this present paper, we push further the  study of $2$-to-$1$ mappings, particularly, over finite fields with characteristic $2$ (binary case being the most interesting for applications). Firstly,  we completely determine  $2$-to-$1$ polynomials with degree $5$ over $\gf_{2^n}$ using the well known Hasse-Weil bound. Besides, we consider $2$-to-$1$ mappings with few terms, mainly trinomials and quadrinomials. Using the multivariate method and the resultant of two polynomials, we present  two classes of $2$-to-$1$ trinomials, which explain all the examples of $2$-to-$1$ trinomials of the form $x^k+\beta x^{\ell} + \alpha x\in\gf_{{2^n}}[x]$ over $\gf_{{2^n}}$ with $n\le 7$,  and derive twelve classes of  $2$-to-$1$  quadrinomials with trivial coefficients over $\gf_{2^n}$.
\end{abstract}

   \begin{IEEEkeywords}
	 $2$-to-$1$ mapping, Polynomial of low degree, Trinomial, Quadrinomial, Symmetric cryptography.
\end{IEEEkeywords}

\section{Introduction}
Let $\gf_{2^n}$ be the finite field of order $2^n$ and $f$ be a mapping from $\gf_{2^n}$ to itself. $f$ is said to be a $2$-to-$1$ mapping if $\# f^{-1}(a) \in \{ 0,2\},$ for any $a\in\gf_{2^n}$.  $2$-to-$1$ mappings over finite fields in even characteristic have wide applications in symmetric cryptography, in particular in the construction of APN functions, bent functions, semi-bent functions and so on. For example, $2$-to-$1$ mappings over finite fields with characteristic $2$ allow to construct bent Boolean functions in bivariate representation from the so-called class $\mathcal{H}$ introduced by Carlet and Mesnager \cite{CM2011}.  In addition, $2$-to-$1$ mappings over $\gf_{2^n}$ can also determine semi-bent Boolean functions in bivariate representation from the Maiorana-MaFarland class \cite{MQ2019}. For more applications of $2$-to-$1$ mappings over finite fields, we refer to \cite[Section 6]{MQ2019}.

 Very recently, Mesnager and Qu \cite{MQ2019} provided a systematic study of $2$-to-$1$ mappings over finite fields with arbitrary characteristic including characterizations by the Walsh transform, several constructions (an AGW-like criterion, those from permutation polynomials or linear translators), some classical classes of polynomials (linearized polynomials, monomials, low degree polynomials, etc) and many explicit applications of $2$-to-$1$ mappings. In this paper\footnote{ Some parts of the present paper have been accepted in the proceedings of the conference IWSDA 2019, mainly the determination of $2$-to-$1$ mappings with degree $5$ and two classes of $2$-to-$1$ trinomials and four classes of $2$-to-$1$  quadrinomials. However, due to the limit of length, only sketch  of the proofs have been included in the proceedings paper.} we push  further the study  initiated in \cite{MQ2019} and we focus on even characteristic (binary case being the most interesting for applications). More specifically, in \cite{MQ2019}, the authors determined completely  the $2$-to-$1$ mappings of degree at most 4 over finite fields. In this paper, we firstly consider $2$-to-$1$ polynomials of degree $5$ over $\gf_{2^n}$. Our method is based on the Hasse-Weil bound (see e.g. \cite{Houbook2018,S1993}), which has been used recently in the study of permutation polynomials (see e.g. \cite{Bartoli2018,Hou2018}). Next, we focus on $2$-to-$1$ mappings with few terms. The characterization of $2$-to-$1$ monomials is trivial. The authors \cite{MQ2019} presented four classes of $2$-to-$1$ binomials of the form $x^k+x^{\ell}$ thanks to the known results on hyperoval sets (see Proposition \ref{2-1-binomials}). In the present paper, we investigate by  MAGMA all $2$-to-$1$ binomials of the form $f(x) = x^{k}+ \alpha x^{\ell} $ on $\gf_{{2^n}}$, where $k\ge \ell$, $\alpha\in\gf_{{2^n}}^{*}$ and $n\le 7$. It turns out that for $n$ odd, up to the o-equivalence (see Definition \ref{o-equivalence}), these experiment examples can be explained by the four classes of $2$-to-$1$ binomials given in Proposition \ref{2-1-binomials}. For the case when $n$ is even, there is only one class of $2$-to-$1$ binomials that is $x^{2^n-2}+x$ over $\gf_{{2^n}}$.  For $2$-to-$1$ trinomials and quadrinomials, we derive, up to the QM-equivalence (see Definition \ref{m-equivalence}),  two classes of $2$-to-$1$ trinomials of the form $x^k+\beta x^{\ell} + \alpha x\in\gf_{{2^n}}[x]$, which explain all the experiment examples of $2$-to-$1$ trinomials  over  $\gf_{{2^n}}$ with $n\le 7$. We also present twelve classes of  $2$-to-$1$ quadrinomials of the form $x^k+x^{\ell} +x^d+x$ over $\gf_{2^n}$.  Our method in this part uses the multivariate method introduced by Dobbertin \cite{Dobbertin2002} and the key point is to determine the number of solutions of equations with a high degree. In addition, we shall use an important tool (which is the resultant of two polynomials) to treat the case of quadrinomials.

The remainder of this paper is organized as follows. Section 2 introduces some basic notions which will be used in the manuscript.  Based on the Hasse-Weil bound, Section 3 determines completely the $2$-to-$1$ mappings with degree $5$ over $\gf_{2^n}$. In Section 4, we consider $2$-to-$1$ polynomials with few terms, including  mainly experiment examples obtained by MAGMA and two classes of $2$-to-$1$ trinomials over $\gf_{{2^n}}$.  In Section 5, we present twelve classes of $2$-to-$1$  quadrinomials over $\gf_{2^n}$.  Finally, Section 6 is a conclusion.
 Throughout this paper, for any $m \mid n$, we assume $\tr_{m}^{n}$ denotes the trace function from $\gf_{{2^n}}$ to $\gf_{2^m}$, i.e., $\tr_m^n(x) = x+x^{2^m}+\cdots+x^{2^{\left(\frac{n}{m}-1\right)m}}$ for any $x\in\gf_{{2^n}}$. Particularly, when $m=1$, we use $\tr_{2^n}$ to denote the absolute trace function over $\gf_{2^{n}}$, i.e., $\tr_{2^{n}}(x) = x+x^2+\cdots+x^{2^{n-1}}$ for any $x\in\gf_{2^{n}}$.  The algebraic closure of $\gf_{2^{n}}$ is denoted by $\overline{\gf}_{2^n}$. For any sets $E$, $\# E$ denotes the cardinality of $E$.

\section{Preliminaries}

In this section, we introduce some basic notions on the Hasse-Weil bound as well as some known results  concerning the solutions of equations with low degree (quadratic, cubic and quartic). We also recall  the resultant of two polynomials which will be useful in our subsequent proofs. Finally, we introduce and recall two equivalences to study $2$-to-$1$ polynomials.

\subsection{The Hasse-Weil bound}

In this subsection, we recall some well known results on algebraic curves and algebraic function fields, mainly the Hasse-Weil bound.  These classical results can be found in most of the textbooks on algebraic curves and algebraic function fields.

\begin{Lemma}
\cite[Hasse-Weil bound]{Houbook2018,S1993}
\label{Hasse-weil}
Let $G(X,Y)$ be an absolutely irreducible polynomial in $\gf_{q}[X,Y]$ of degree $d$ and let $\# V_{\gf_{q}^2}(G)$ be the number of zeros of $G$. Then 
$$ \left| \# V_{\gf_{q}^2}(G)-q \right| \le 2gq^{1/2}, $$
where $g$ is the genus of the function field $\gf_{2^n}(\mathbb{X}, \mathbb{Y} )/ \gf_{2^n}$ and $ \mathbb{X}, \mathbb{Y} $ are transcendentals over $\gf_{2^n}$ with $G(\mathbb{X}, \mathbb{Y} ) = 0$. 	
\end{Lemma}

Let $F/K$ be a function field and $K$ be perfect. Let $g$ denote the genus of $F/K$. Then we have the following upper bound on  the genus.

\begin{Lemma}
	\cite{Houbook2018,S1993}
	\label{genus-bound}
	Let $ F = K(\mathbb{X}, \mathbb{Y} ) $, where  $ \mathbb{X}, \mathbb{Y} $ are transcendentals over $K$. Then the genus of the  function field $F/ K$ satisfies:
	$$ g\le ([F: K(\mathbb{X})] -1 )([F: K(\mathbb{Y})] -1). $$
\end{Lemma}

Given two plane curves $\mathcal{A}$ and $\mathcal{B}$ and a point $P$ on the plane, the \emph{intersection number} $I(P,\mathcal{A}\cap \mathcal{B})$ of $\mathcal{A}$ and $\mathcal{B}$ at the point $P$ is defined by seven axioms. We do not include its precise and long definitions here. For more details, we refer to \cite{HKT2008}.
\begin{Lemma}
 \cite[B\'{e}zout's Theorem]{HKT2008}
	\label{Btheorem}
	Let $\mathcal{A}$ and $\mathcal{B}$ be two projective plane curves over an algebraically closed field $K$, having no component in common. Let $A$ and $B$ be the polynomials associated with $\mathcal{A}$ and $\mathcal{B}$ respectively. Then 
	$$\sum_{P}I(P,\mathcal{A}\cap \mathcal{B})=(\deg A) (\deg B),$$
	where the sum runs over all points in the projective plane $\mathrm{PG}(2,K)$.
\end{Lemma}

\subsection{Solutions of equations with low degree}

In this subsection, we introduce some known lemmas about the solutions of  some equations with low degree (quadratic, cubic, quartic), which will be used the proofs of our results.

\begin{Lemma}
	\cite{LN1997}
	\label{lem_deg_2}
	Let $u,v\in\gf_{2^n}$ and $u\neq0$. Then the quadratic equation $x^2+ux+v=0$ has solutions in $\gf_{2^n}$ if and only if $\tr_{2^n}\left(\frac{v}{u^2}\right)=0$.
\end{Lemma}

\begin{Lemma}
	\cite{SZM1967}
	\label{lem_deg_3}
	Let $a, b \in \gf_{2^n}$, where $b\neq0$. Then the cubic equation $x^3+ax+b=0$ has a unique solution in $ \gf_{2^n}$ if and only if $\tr_{2^n}\left(\frac{a^3}{b^2}+1\right) \neq 0$.
\end{Lemma}


If $h(x)$ is a quartic polynomial over $\gf_{2^n}$ which factors as a product of two linear factors times an irreducible quadratic, we write $h=(1,1,2)$; if $h(x)$ is a cubic irreducible polynomial over $\gf_{2^n}$, we write $h=(3)$.  In \cite{LW1972}, P. A. Leonard and K. S. Williams characterized the factorization of a quartic polynomial over $\gf_{2^n}$ as follows.

\begin{Lemma}
	\cite{LW1972}
	\label{lem_deg_4}
	Let $f(x)=x^4+a_2x^2+a_1x+a_0$ with $a_i\in\gf_{2^n}$ and $a_0a_1\neq0$. Let $f_1(y)=y^3+a_2y+a_1$ and $r_1,r_2,r_3$ denote roots of $f_1(y)=0$ when they exist in $\gf_{2^n}$. Set $w_i=a_0\frac{r_i^2}{a_1^2}$. Then the factorization of $f(x)$ over $\gf_{2^n}$ is characterized as follows:
	\begin{enumerate}[(1)]
		\item $f=(1,1,1,1)$ if and only if $f_1=(1,1,1)$ and $\tr_{2^n}(w_1)=\tr_{2^n}(w_2)=\tr_{2^n}(w_3)=0$;
		\item $f=(2,2)$ if and only if $f_1=(1,1,1)$ and $\tr_{2^n}(w_1)=0$, $\tr_{2^n}(w_2)=\tr_{2^n}(w_3)=1$;
		\item $f=(1,3)$ if and only if $f_1=(3)$;
		\item $f=(1,1,2)$ if and only if $f_1=(1,2)$ and $\tr_{2^n}(w_1)=0$;
		\item $f=(4)$ if and only if $f_1=(1,2)$ and $\tr_{2^n}(w_1)=1$.
	\end{enumerate} 
\end{Lemma}

\subsection{Resultant of polynomials}

In this subsection, we  recall some basic facts about the resultant of two polynomials. Given two non-zero polynomials of degrees $n$ and $m$ respectively
$$ u(x) = a_mx^m+a_{m-1}x^{m-1}+\cdots+a_0, ~~~ v(x) = b_nx^n+b_{n-1}x^{n-1}+\cdots+b_0 $$
with $a_m\neq0, b_n\neq0$ and coefficients in a field or in an integral domain $\mathbb{R}$, their resultant $\mathrm{Res}(u,v)\in \mathbb{R}$ is the determinant of the following matrix:
$$  \begin{pmatrix}
a_m & a_{m-1} & \cdots & \cdots & \cdots & \cdots & a_0 & 0 & 0 & 0 \\
0 & a_m & a_{m-1} & \cdots & \cdots & \cdots & \cdots &  a_0 & 0 & 0 \\
&  & \ddots & \ddots & & & &  &   \ddots &      \\
0 & 0 & 0 & a_m & a_{m-1} & \cdots & \cdots & \cdots & \cdots & a_0 \\
b_n & b_{n-1} & \cdots & \cdots & b_0 & 0 & 0 & 0 & 0 & 0\\
& \ddots & \ddots &  & & \ddots & &  &   &      \\ 
0   & 0 & 0 & 0 & 0 &  b_n & b_{n-1} &  \cdots &   \cdots &  b_0    \\
\end{pmatrix}.
$$
For a field $K$ and two polynomials $F(x,y), G(x,y) \in K[x,y]$, we use $ \mathrm{Res}(F,G,y)$ to denote the resultant  of $F$ and $G$ with respect to $y$. It is the resultant of $F$ and $G$ when considered as polynomials in the single variable $y$. In this case, $ \mathrm{Res}(F,G,y)\in K[x]$ belongs in the ideal generated by $F$ and $G$, and thus any $a,b$ satisfying $F(a,b)=0$ and $G(a,b)=0$ is such that $\mathrm{Res}(F,G,y)(a)=0$ (see \cite{LN1997}). 

\subsection{O-equivalence and QM-equivalence between two 2-to-1 mappings}

A permutation polynomial $f$ over $\gf_{{2^n}}$ is called an oval-polynomial (for short o-polynomial) if $f(0)=0$, and for each $s\in\gf_{{2^n}}$, $$f_s(x)=\frac{f(x+s)+f(s)}{x}$$
is a permutation polynomial.  It is well known that there is a close relation between o-polynomials  and $2$-to-$1$ mappings as follows.

\begin{Lemma}
	A polynomial $f$ from $\gf_{{2^n}}$ to itself with $f(0)=0$ is an o-polynomial if and only if $f_a(x) := f(x)+ax$ is $2$-to-$1$ for every $a\in\gf_{{2^n}}^{*}$. 
\end{Lemma}

Each o-polynomial defines an (hyper-)oval. And  each hyperoval defines $2^n +2$ o-polynomials.
Two o-polynomials are called (projectively) o-equivalent\footnote{Note that  for a special type of bent functions, so-called Niho bent functions there is a
general equivalence relation called o-equivalence which is induced from the
equivalence of o-polynomials.}, if they define equivalent
hyperovals. Hyperovals being called equivalent if they are mapped to each other by collineations (i.e. permutations mapping lines to lines).  For example, for the following o-monomials $$x^k, x^{\frac{1}{k}}, x^{1-k}, x^{\frac{1}{1-k}}, x^{\frac{k}{k-1}}, x^{\frac{k-1}{k}}, $$ are o-equivalent to each other. We can naturally define an equivalent relation between $2$-to-$1$ (polynomial) mappings  as follows.

\begin{Def}
	\label{o-equivalence}
	Let $f_a(x) = f(x)+ax$ and $g_a(x) =g(x)+ax $ be $2$-to-$1$ (polynomial) mappings where $a\in\gf_{{2^n}}^{*}$. Then $f_a$ and $g_a$ are said to be o-equivalent if the corresponding o-polynomials $f$ and $g$ are equivalent.
\end{Def}

The o-equivalence between two 2-to-1 mappings will play an important role in our classification of $2$-to-$1$ binomials. Namely, the o-equivalence plays a major role in explaining the experiment results of $2$-to-$1$ binomials.
In particular,  under the o-equivalence, Proposition 4.2 can explain all experiment results on $2$-to-$1$ binomials over $\gf_{2^n}$ with $n\le7$ and odd. 

Next, we recall another equivalence between polynomials. Let $f(x), g(x), h(x)\in\gf_{q}[x]$ with $f(x)=g(h(x))$.  As we all know, if $h(x)$ permutes $\gf_q$, then  $f$  is a $2$-to-$1$ polynomial over $\gf_q$ if and only if $g(x)$ is $2$-to-$1$ polynomial over $\gf_q$.

 It is trivial that a monomial $x^d$ is a permutation polynomial over $\gf_q$ if and only if $\mathrm{gcd}(d, q-1)=1.$  Let $f$ and $g$ be two polynomials in $\gf_q[x]$ satisfying that $f(x) = a g(bx^d)$, where $a,b\in \gf_q^\ast$ and $1 \le d \le q-1$ is an integer such that $\mathrm{gcd}(d,q-1)=1$. Then $f$ is $2$-to-$1$ if and only if so is $g$.  Consequently,  we recall the following notion of \emph{QM-equivalence} \cite{WYD2017,LQC2017}.
\begin{Def}
	\cite{WYD2017}
	\label{m-equivalence}
	Two polynomials $f(x)$ and $g(x)$ in $\mathbb{F}_{q}[x]$ are said to be quasi-multiplicative (QM) equivalence if there exists an integer $1 \le d \le q-1$ with $\gcd(d,q-1)=1$ and $f(x)=a g(b  x^d)$ for some nonzero elements $a, b$ in $\gf_q$.
\end{Def}

Using the QM-equivalence, we can simplify the experiment results on $2$-to-$1$ trinomials (resp. quadrinomials). As indicated in Table II, there are for example only two $2$-to-$1$ trinomials of the form $x^k+\beta x^{\ell}+\alpha x$ over $\gf_{2^6}$ up to the QM-equivalence. In addition, we can avoid getting equivalent $2$-to-$1$ polynomials having the same terms. 

\section{ $2$-to-$1$ mappings with degree $5$ over $\gf_{2^n}$}

In this section, we completely determine the $2$-to-$1$ mappings with degree $5$ over $\gf_{2^n}$. Clearly, for any polynomials $f(x)\in\gf_{p^n}[x]$ with degree $d$, is $2$-to-$1$ over $\gf_{p^n}$ if and only if so is $g(x)=af(x+b)+c$, where $a\in\gf_{p^n}^{*}$ and $b,c\in\gf_{p^n}$. Hence, it is suffisant  to consider $f(x)\in\gf_{p^n}[x]$ with \emph{normalized form}, i.e., $f(x)$ is monic, $f(0)=0$, and when $\gcd(p,d)=1$, the coefficient of $x^{d-1}$ is $0$. That is to say, in this part, we suffice to consider $f(x) = x^5+a_3x^3+a_2x^2+a_1x,$ where $a_3,a_2,a_1\in\gf_{2^n}$ since $\gcd(2,5)=1$.

\begin{Th}
	Let $n>9$ and $f(x) = x^5+a_3x^3+a_2x^2+a_1x,$ where $a_3,a_2,a_1\in\gf_{2^n}$. Then $f(x)$ is not $2$-to-$1$ over $\gf_{2^n}$.
\end{Th}

\begin{proof}
	We assume that $f(x)$ is $2$-to-$1$ over $\gf_{2^n}$. According to the definition,  $f(x+a)+f(a)=0$ has exactly two solutions $x=0$  and $x=x_0\in\gf_{2^n}^{*}$ for any $a\in\gf_{2^n}$.  In addition, 	for any $a \in\gf_{2^n}$,  we have
	\begin{eqnarray*}
		&  & f(x+a) + f(a) \\
		& = & x^5+ax^4+a_3x^3+\left(a_3a+a_2\right)x^2+\left(a^4+a_3a^2+a_1\right)x.
	\end{eqnarray*} 
	Thus for any $a\in\gf_{2^n}$,
	\begin{equation}
	\label{case_4_1}
	x^4+ax^3+a_3x^2+\left(a_3a+a_2\right)x+a^4+a_3a^2+a_1=0
	\end{equation}
	has exactly one solution in $\gf_{2^n}^{*}$. 
	
  In the following, we assume that $a\neq0 $ and $a^4+a_3a^2+a_1\neq0$. Let $$F(x)=x^4+ax^3+a_3x^2+\left(a_3a+a_2\right)x+a^4+a_3a^2+a_1$$ and $\alpha^2=a_3+\frac{a_2}{a}.$
	Then $$x^4F\left(\frac{1}{x}+\alpha\right) = \tilde{A_0}x^4 + \tilde{A_2}x^2 + ax + 1,$$ where 
	$\tilde{A_0}=\left(\alpha^2+a^2\right)\left(\alpha^2+a^2+a_3\right)+a_1$ and $\tilde{A_2}=a\alpha + a_3$. We also assume that $\tilde{A_0}\neq 0$. Thus Eq. (\ref{case_4_1}) has exactly one solution in $\gf_{2^n}^{*}$ if and only if
	\begin{equation}
	\label{d4_2}
	x^4+A_2x^2+A_1x+A_0=0
	\end{equation}
does, where $A_2=\frac{\tilde{A_2}}{\tilde{A_0}}, A_1=\frac{a}{\tilde{A_0}},$ $A_0=\frac{1}{\tilde{A_0}}$. Let $f_1(y)=y^3+A_2y+A_1$. According to Lemmas \ref{lem_deg_3} and \ref{lem_deg_4}, if $x^4+A_2x^2+A_1x+A_0=0$ has exactly one solution in $\gf_{2^n}^{*}$, then $f_1 = (3)$ and
	\begin{equation}
	\label{hengdeng}
	\tr_{2^n}\left(\frac{A_2^3}{A_1^2}+1\right)=0
	\end{equation}
	for any  $a\in\gf_{2^n}^{*}$ satisfying  $a^4+a_3a^2+a_1\neq0$ and $\tilde{A_0}\neq 0$. Indeed,
	\begin{eqnarray*}
		\tr_{2^n}\left(\frac{A_2^3}{A_1^2}+1\right)&=&\tr_{2^n}\left(\frac{\left(a_3a^2+a_2a+a_3^2\right)^3}{\left(a_3^2a^2+a_2^2+a^6\right)\left(a_2^2+a^6\right)+a_1^2a^4}+1\right) \\
		&\triangleq& \tr_{2^n}\left(\frac{L_1(a)}{L_2^2(a )}\right),
	\end{eqnarray*}
	where $L_1(a)=a^{12}+a_3^2a^8+a_3^3a^6+a_2a_3^2a^5+\left(a_1^2+a_3^4+a_2^2a_3\right)a^4+a_2^3a^3+a_3^5a^2+a_2a_3^4a+a_2^4+a_3^6$ and $L_2(a)=a^6+a_3a^4+a_1a^2+a_2a_3a+a_2^2$. It should be noted that $\tilde{A_0}\neq 0$ is equivalent to $L_2(a)\neq 0$. 
	
	Assume that $L_1(Y) = Y^{12}+a_3^2Y^8+a_3^3Y^6+a_2a_3^2Y^5+\left(a_1^2+a_3^4+a_2^2a_3\right)Y^4+a_2^3Y^3+a_3^5Y^2+a_2a_3^4Y+a_2^4+a_3^6\in\gf_{2^n}[Y]$ and $L_2(Y) =  Y^6+a_3Y^4+a_1Y^2+a_2a_3Y+a_2^2 \in\gf_{2^n}[Y]$. Let 
	$$G(X,Y) = L_2^2(Y) \left(X^2+X\right)+L_1(Y)$$
	and 
	$$ V_n(G) = \{ (x,y)\in\gf_{2^n}^2 ~|~ G(x,y) = 0  \}.  $$ 
	Then $\deg G = 14$. Together with Eq. (\ref{hengdeng}) and that there exist at most $2$ elements $a\in\gf_{2^n}^{*}$ such that $a^4+a_3a^2+a_1=0$, as well as that at most $6$  elements $a\in\gf_{2^n}^{*}$ such that $L_2(a)=0$, we have 
	\begin{equation}
	\label{V_n}
	\# V_n(G)  \ge 2\left( 2^n-9 \right).
	\end{equation}
	
	If $G(X,Y)$ is  irreducible over $\overline{\gf}_{2^n}$, let $ \mathbb{X}, \mathbb{Y} $ be transcendentals over $\gf_{2^n}$ with $G(\mathbb{X}, \mathbb{Y} ) = 0$. Then by Lemma \ref{genus-bound}, the functional fields $\gf_{2^n}(\mathbb{X}, \mathbb{Y} )/ \gf_{2^n}$ has genus
	$$  g\le ([\gf_{2^n}(\mathbb{X}, \mathbb{Y} ): \gf_{2^n}(\mathbb{X})] -1 )([\gf_{2^n}(\mathbb{X}, \mathbb{Y} ): \gf_{2^n}(\mathbb{Y})] -1) \le (12-1)(2-1)=11. $$
	Then by the Hasse-Weil bound, i.e., Lemma \ref{Hasse-weil}, we have 
	$$  \# V_n(G) \le 2^n+1+2g2^{n/2}\le 2^n+1+ 22\cdot2^{n/2}<  2\left( 2^n-9 \right), $$
	when $n>9$, which is contradictory with (\ref{V_n}). 
	
	Therefore, $G(X,Y)$ is not irreducible over $\overline{\gf}_{2^n}$ and we assume that $G(X,Y) =  sG_1(X,Y)G_2(X,Y),$ where $s\in\gf_{2^n}^{*}$, $G_1, G_2 \in \overline{\gf}_{2^n}[X,Y]$ are irreducible and $\deg_X G_1 = \deg_X G_2 = 1.$ If $G_1\notin \gf_{2^n}[X,Y]$, choose $\sigma\in\mathrm{Aut} (\overline{\gf}_{2^n}/\gf_{2^n})$ such that $\sigma(G_1)\neq G_1$. Then $\sigma(G_1)=G_2$. Assume that $(x,y)\in V_{n}(G)$. Then $(x,y)\in V_{n}(G_1)$ or $V_{n}(G_2)$, say $(x,y)\in V_{n}(G_1).$ Then $(x,y)=(\sigma(x),\sigma(y))\in V_{n}(\sigma (G_1))$. Hence $(x,y)\in V_{n}(G_1)\cap V_{n}(\sigma (G_1))$ and we have 
	$$ V_n(G) \subset V_n(G_1)\cap V_n(\sigma(G_1)). $$
	Thanks to B\'{e}zout's Theorem, i.e., Lemma \ref{Btheorem},
	$$ \#V_n(G)\le \# \left( V_n(G_1) \cap V_n(\sigma(G_1))  \right) \le (\deg G_1)^2 \le 49,$$
	which is also contradictory with (\ref{V_n}). Thus $G_1, G_2\in\gf_{2^n}[X,Y]$. Namely, there exists some $L(Y)\in\gf_{2^n}[Y]$ such that 
	$$X^2+X+\frac{L_1(Y)}{L_2(Y)^2} = \left(X+ \frac{L(Y)}{L_2(Y)} \right)   \left( X+ 1 + \frac{L(Y)}{L_2(Y)} \right). $$
	Hence, 
	\begin{equation}
	\label{LY}
	L(Y)L_2(Y)+L(Y)^2 = L_1(Y).
	\end{equation}
Obviously, the degree of $L(Y)$ is $6$.	Assume that $L(Y) = \ell_0 Y^6+\ell_1Y^5+\ell_2Y^4+\ell_3Y^3+\ell_4Y^2+\ell_5Y+\ell_6\in\gf_{2^n}[Y].$ After comparing the coefficients of Eq. (\ref{LY}), we have 
	\begin{subequations} 
		\small
		\renewcommand\theequation{\theparentequation.\arabic{equation}}     
		\begin{empheq}[left={\empheqlbrace\,}]{align}
		&~D_{12}: 	\ell_0^2+\ell_0=1 \label{D12}  \\
		&~D_{11}:  \ell_1=0  \label{D11}  \\
		&~ D_{10}:  \ell_2+a_3\ell_0 =0  \label{D10} \\
		&~ D_{9}:  \ell_3=0, \label{D9}  \\
		&~D_{8}:  \ell_4+a_3\ell_2+a_1\ell_0+\ell_2^2=a_3^2, \label{D8}  \\
		&~ D_{7}:  \ell_5+a_2a_3\ell_0 = 0, \label{D7}  \\
		&~ D_{6}:  \ell_6+ a_3\ell_4+a_1\ell_2+a_2^2\ell_0   =a_3^3, \label{D6}  \\
		&~ D_5: a_3\ell_5+a_1a_3\ell_2=a_2\ell_3^2, \label{D5}  \\
		&~ D_4: a_3\ell_6+a_1\ell_4+a_2\ell_2+\ell_4^2=a_1^2+a_3^4+a_2^2a_3, \label{D4} \\
		&~ D_3: a_1\ell_5+a_2a_3\ell_4 = a_2^3, \label{D3}\\
		&~ D_2: a_1\ell_6+a_2a_3\ell_5+a_2^2\ell_4+\ell_5^2 = a_3^5, \label{D2} \\
		&~ D_{1}:   a_2a_3\ell_6+a_2^2\ell_5=  a_2a_3^4, \label{D1}  \\
		&~ D_{0}:  a_2^2\ell_6+\ell_6^2 = a_2^4+a_3^6. \label{D0} 
		\end{empheq}
	\end{subequations}
	In the above equation system, $D_i$ denotes that the equation in the same row  is from comparing the coefficient of degree $i$.  From (\ref{D12})-(\ref{D6}), we obtain $\ell_0^2+\ell_0=1$, $\ell_1=0, \ell_2= a_3\ell_0, \ell_3=0, \ell_4= a_1\ell_0, \ell_5 = a_2a_3\ell_0$ and $\ell_6 = a_2^2\ell_0+a_3^3$. Together with (\ref{D0}), we have $$a_2^2a_3^3=0.$$ Then it follows from (\ref{D7}) that $\ell_5=0$, and from  (\ref{D3}), we  have $a_2=0$. Thus in the following, we suffice to consider the case $a_2=0$. When $a_2=0$, from (\ref{D0}), (\ref{D2}) and $\ell_6=a_3^3$,  we have $a_1a_3^3=a_3^5$, and thus $a_3=0$ or $a_1=a_3^2\neq0$. 
	
	When  $ a_2=a_3=0$, $f(x)=x^5+a_1x$ and Eq. (\ref{d4_2}) becomes 
	\begin{equation}
	\label{case_3_1}
	x^4+\frac{a}{a^4+a_1}x+\frac{1}{a^4+a_1} = 0,
	\end{equation}
  where $a\in\gf_{2^n}\backslash\{a_1^{1/4}\}$.
Let $f_1(y)=y^3+\frac{a}{a^4+a_1}.$ Then if $f(x)$ is $2$-to-$1$, $f_1(y)$ must be irreducible from Lemma \ref{lem_deg_4}. However, it is clear that there exist some $a\in\gf_{2^n}\backslash\{a_1^{1/4}\}$ such that $y^3=\frac{a}{a^4+a_1}$ has $1$ or $3$ solutions in $\gf_{2^n}$, which means $f_1(y)$ can not be irreducible.  Thus $f(x)$ is not $2$-to-$1$ in the case.
	
	When $a_1=a_3^2\neq0$, Eq. (\ref{case_4_1}) becomes 
	\begin{equation}
	\label{eq_1}
	x^4+ax^3+a_3x^2+a_3ax+a^4+a_3a^2+a_3^2=0.
	\end{equation}
	Let $a=0$. Then from Eq. (\ref{eq_1}), we get $x^4+a_3x^2+a_3^2 =0 $, having $0$ or $2$  solutions in $\gf_{2^n}^{*}$, which is contrary with that Eq. (\ref{case_4_1}) has only one solution in $\gf_{2^n}^{*}$ for any $a\in\gf_{2^n}$.   

	Therefore when $n>9$,  $f(x)$ is not $2$-to-$1$ over $\gf_{2^n}$, which completes the  proof.
\end{proof}

\begin{Rem}
	As for the case $3\le n\le 9$, there exist some $2$-to-$1$ mappings with the form of $f(x) = x^5+a_3x^3+a_2x^2+a_1x,$ where $a_3,a_2,a_1\in\gf_{2^3}$. We obtain them by MAGMA and list them in Table I, where $\gamma$ is a primitive element in $\gf_{2^3}$.
	 \begin{table}[!htbp]
		\label{n=3}
		\caption{$\left(a_3,a_2,a_1\right)\in\gf_{2^3}^{3}$ such that $f(x)$ is $2$-to-$1$}
		\centering
		\begin{tabular}{c|c||c|c||c|c}	
			\toprule
			No. & $\left(a_3,a_2,a_1\right)$ & No.  &  $\left(a_3,a_2,a_1\right)$ & No. & $\left(a_3,a_2,a_1\right)$ \\
			\midrule
	        $1$ & $\left(1,\gamma,\gamma^5\right)$ & $2$ & $\left(1,\gamma^2,\gamma^3\right)$ & $3$ & $\left(1,\gamma^4,\gamma^6\right)$ \\
	        $4$ & $\left(\gamma,1,\gamma^5\right)$ & $5$ & $\left(\gamma,\gamma^2,\gamma\right)$ & $6$ & $\left(\gamma,\gamma^6,1\right)$ \\
	        $7$ & $\left(\gamma^2,1,\gamma^3\right)$ & $8$ & $\left(\gamma^2,\gamma^4,\gamma^2\right)$ & $9$ & $\left(\gamma^2,\gamma^5,1\right)$ \\
	        $10$ & $\left(\gamma^3,\gamma^2,\gamma^4\right)$ & $11$ & $\left(\gamma^3,\gamma^3,\gamma^2\right)$ & $12$ & 
	        $\left(\gamma^3,\gamma^5,\gamma^5\right)$ \\
	        $13$ & $\left(\gamma^4,1,\gamma^6\right)$ & $14$ & $\left(\gamma^4,\gamma,\gamma^4\right)$ & $15$ & 
	        $\left(\gamma^4,\gamma^3,1\right)$ \\
	        $16$ & $\left(\gamma^5,\gamma,\gamma^2\right)$ & $17$ & $\left(\gamma^5,\gamma^5,\gamma\right)$ & $18$ & $\left(\gamma^5,\gamma^5,\gamma^6\right)$ \\
	        $19$ & $\left(\gamma^6,\gamma^3,\gamma^3\right)$ & $20$ & $\left(\gamma^6,\gamma^4,\gamma\right)$ & $21$ &
	        $\left(\gamma^6,\gamma^6,\gamma^4\right)$\\
	        $22$ & $\left(a_3,0,0\right)$, $a_3\in\gf_{2^3}^{*}$ & $23$ & $\left(0,a_2,0\right)$, $a_2\in\gf_{2^3}^{*}$  & & \\
			\bottomrule
		\end{tabular}
	\end{table}
\end{Rem}

\section{$2$-to-$1$ monomials, binomials and trinomials over  $\gf_{2^{n}}$}

\subsection{$2$-to-$1$ monomials and binomials}

In \cite{MQ2019}, the authors gave the characterization of $2$-to-$1$ monomials over $\gf_q$ with arbitrary characteristic. 
\begin{Prop}
	\cite{MQ2019}
	Let $f(x)=ax^d$ be a monomial over $\gf_q$, where $a\neq0$. Then $f(x)$ is $2$-to-$1$ over $\gf_q$ if and only if $\gcd(d,q-1)=2$.
\end{Prop}

As for binomials, in \cite{MQ2019}, the authors presented the following proposition about $2$-to-$1$ binomials of the form $x^k+x$ thanks to the known results about hyperoval sets and o-polynomials.   
\begin{Prop}
	\cite{MQ2019}
	\label{2-1-binomials}
	Let $n$ be odd. Then $f(x) = x^k+x$ is $2$-to-$1$ over $\gf_{{2^n}}$ if one of the following case holds.
	\begin{enumerate}
		\item $k=2$ (the Singer case);
		\item $k=6$ (the Segre case);
		\item $k=2^{\sigma}+2^{\pi}$ with $\sigma=\frac{n+1}{2}$ and $4\pi\equiv 1\pmod n$ (the Glynn I case);
		\item $k=3\cdot 2^{\sigma}+4$ with $\sigma=\frac{n+1}{2}$ (the Glynn II case).
	\end{enumerate} 
\end{Prop}  

 In addition, it is well known that the o-monomials $$x^k, x^{\frac{1}{k}}, x^{1-k}, x^{\frac{1}{1-k}}, x^{\frac{k}{k-1}}, x^{\frac{k-1}{k}}, $$  are o-equivalent to each other. Namely, for such $k$ in Proposition \ref{2-1-binomials},
  $$x^{\frac{1}{k}}+x, x^{1-k}+x, x^{\frac{1}{1-k}}+x, x^{\frac{k}{k-1}}+x, x^{\frac{k-1}{k}}+x$$  are also $2$-to-$1$. Actually, they are o-equivalent as we define in the previous section.     

More generally, for binomials with the form $f(x) = x^{k}+ \alpha x^{\ell} $ on $\gf_{{2^n}}$, where $k\ge \ell$. Clearly,  when $\gcd\left( \ell, 2^n-1 \right) = 1$, then there exists some positive integer $\ell^{-1}$ such that $\ell\cdot \ell^{-1} \equiv 1 \pmod {2^n-1}$ and $f(x^{\ell^{-1}}) = x^{k\cdot \ell^{-1}} +\alpha x $ is $2$-to-$1$ if and only if so $f(x)$ is. If $\alpha^{\frac{2^n-1}{\gcd\left( k-\ell, 2^n-1  \right)}} =1, $ then there exists some $\beta\in\gf_{2^{n}}$ such that $\alpha=\beta^{k-\ell}$  and $f(\beta x) =\beta^k \left( x^k + x^{\ell} \right) $ is  $2$-to-$1$ if and only if so $f(x)$ is.

By MAGMA, we looked for all the possible $2$-to-$1$ in the form $f(x) = x^{k}+ \alpha x^{\ell} $ on $\gf_{{2^n}}$, where $k\ge \ell$, $\alpha\in\gf_{{2^n}}^{*}$ and $n\le 7$ under o-equivalence and QM-equivalence.  The experiment results show that for $n$ odd, there are only four classes of $2$-to-$1$ binomials as Proposition \ref{2-1-binomials}. However, for $n$ even, there is one class of $2$-to-$1$ binomials, namely, $f(x)=x^{2^n-2}+ x$, whose $2$-to-$1$ property on $\gf_{2^{n}}$ is easy to prove.

\subsection{$2$-to-$1$ trinomials over $\gf_{2^n}$}

In the present subsection, we consider $2$-to-$1$ trinomials of the form $f(x) = x^{k}+ \beta x^{\ell} + \alpha x$, over $\gf_{2^n}$. We firstly obtain some experiment results about $2$-to-$1$ trinomials on $\gf_{{2^n}}$ of the form $f(x) = x^{k}+ \beta x^{\ell} + \alpha x$  by MAGMA (see Table II). 

\begin{table}[!htbp]
	\tiny
	\label{2-1-trinomials}
	\caption{$2$-to-$1$ trinomials of the form $f(x) = x^{k}+ \beta x^{\ell} + \alpha x$ over $\gf_{{2^n}}$, $n\le 7$ }
	\centering
	\begin{tabular}{c c c c c c}	
		\toprule
		$n$ & $k$ & $\ell$ & $\beta$ & $\alpha$ & Ref. \\
		\midrule
		$3$ &  &  &  &  & \\
		$4$ & $12$ & $11$ & $1$ & $\alpha^4+\alpha+1=0$  & Theorem \ref{trinomials-1} \\
		$5$ &  &  &  &  & \\	
		$6$ & $56$ & $55$ & $1$ & $\alpha^8+\alpha+1=0$  & Theorem \ref{trinomials-1} \\	
		$6$ & $13$ & $8$ & $1$ & $\alpha^2+\alpha+1=0$ & Theorem \ref{trinomials-2} \\		
		$7$ &  &  &  & & \\		
		\bottomrule
	\end{tabular}
\end{table}

Note that we only consider the examples that are not QM-equivalent to each other. In Table II, we list the experiment results about $2$-to-$1$ trinomials (not linearized) of the form $f(x) = x^{k}+ \beta x^{\ell} + \alpha x$ over $\gf_{{2^n}}$, where $n\le 7$, $1<\ell < k \le 2^n-2$, $\alpha,\beta \in\gf_{2^{n}}^{*}$.  From Table II, we find that there is not any $2$-to-$1$ trinomials (not linearized) over $\gf_{2^3}$, $\gf_{2^5}$ and $\gf_{2^7}$. In addition, on $\gf_{2^6}$, there are only two $2$-to-$1$ trinomials.  In the following, we generalize the experiment examples to two infinite classes of $2$-to-$1$ trinomials over $\gf_{{2^n}}$.

%

\begin{Th}
	\label{trinomials-1}
	Let $f(x)=x^{2^n-2^m}+x^{2^n-2^m-1}+\alpha x\in\gf_{2^n}[x]$, where $n=2m$ and $\alpha^{2^m}+\alpha+1=0$. Then $f(x)$ is $2$-to-$1$ over $\gf_{2^n}$.
\end{Th}

\begin{proof}
	Let $y=x^{2^m}$ and $b=a^{2^m}$ for any $a\in\gf_{2^n}$. Then $f(x)=\frac{x}{y}+\frac{1}{y}+\alpha x$ and $f(x+a)=\frac{x+a}{y+b}+\frac{1}{y+b}+\alpha(x+a)$. It should be noted that $\frac{1}{0}$ equals $0$ here and we suffice to show that for any $a\in\gf_{2^n}$, $$f(x+a)+f(a)=0,$$
	i.e.,
	\begin{equation}
	\label{trin1}
	\frac{x+a}{y+b}+\frac{1}{y+b}+\alpha x +\frac{a+1}{b}=0
	\end{equation}
	has exactly two solutions in $\gf_{2^n}$. If $a=0$,  then	$b=0$ and Eq. (\ref{trin1}) becomes $\frac{x}{y}+\frac{1}{y}+\alpha x =0$. Obviously, $x=0$ is a solution of the above equation. Now we assume that $x\neq0$. Then the above equation becomes
	\begin{equation}
	\label{trin1-1}
	x+1+\alpha xy=0.
	\end{equation}
	Raising Eq. (\ref{trin1-1}) into its $2^m$-th power, we have 
	\begin{equation}
	\label{trin1-2}
	y+1+\alpha^{2^m}xy=0.
	\end{equation}
  Adding $\alpha^{2^m}\times(\ref{trin1-1})$ and $\alpha\times(\ref{trin1-2})$, we get $y=\alpha^{2^m-1}x+\alpha^{2^m-1}+1$. Furthermore, plugging $y=\alpha^{2^m-1}x+\alpha^{2^m-1}+1$ into Eq. (\ref{trin1-1}), we get $\alpha^{2^m}x^2+\left(\alpha^{2^m}+\alpha+1\right)x+1=0$, which means $x=\alpha^{-2^{m-1}}$ since $\alpha^{2^m}+\alpha+1=0$. Thus Eq. (\ref{trin1}) has exactly two solutions in $\gf_{2^n}$ when $a=0$.

In the following, we assume $a\neq0$. After simplifying Eq. (\ref{trin1}), we obtain
	\begin{equation}
	\label{trin1-3}
	\left(b+\alpha b^2\right)x+(a+1)y +\alpha b xy =0.
	\end{equation} 
Raising Eq. (\ref{trin1-3}) into its $2^m$-th power, we have 
  	\begin{equation}
  \label{trin1-4}
  \left(a+\alpha a^2\right)y+(b+1)x +\alpha^{2^m} a xy =0.
  \end{equation} 
Computing $(\ref{trin1-3})\times\alpha^{2^m}a+(\ref{trin1-4})\times\alpha b$, we get $$Ax=A^{2^m}y,$$
where $A=\alpha^{2^m}ab+\alpha^{2^m+1}ab^2+\alpha b^2+\alpha b.$ It is clear that we only need to claim that $A\neq0$ for any $a\neq0$.  Assume that $A=0$ for some $a\in\gf_{2^n}^{*}$. That is to say, 
\begin{equation}
\label{trin1-5}
\alpha^{2^m}a+\alpha^{2^m+1}ab+\alpha b+\alpha =0.
\end{equation}
Computing $(\ref{trin1-5})+(\ref{trin1-5})^{2^m}$, we have $\alpha^{2^m}+\alpha=0$, which is in conflict with $\alpha^{2^m}+\alpha=1$. Thus $A\neq0$ for any $a\in\gf_{2^n}^{*}$ and it follows that Eq. (\ref{trin1}) has exactly two solutions in $\gf_{2^n}$ when $a\neq0$.

To sum up, Eq. (\ref{trin1}) has exactly two solutions in $\gf_{2^n}$ for any $a\in\gf_{2^n}$. 
\end{proof}

To prove the second class of 2-to-1 trinomials, we firstly recall the well known Dickson polynomials. The $r$-th Dickson polynomial of the first kind $D_r(x,a)\in\gf_q[x]$, where $$D_r(x,a)=\sum_{i=0}^{\lfloor\frac{r}{2}\rfloor}\frac{r}{r-i}\binom{r-i}{i}(-a)^ix^{r-2i},$$ is a permutation polynomial over $\gf_q$ if and only if $\gcd\left(r,q^2-1\right)=1$ \cite{LM1993}. Furthermore, its compositional inverse can be computed by the following lemma.

\begin{Lemma}
	\cite{LQW2018}
	\label{lem_Dick_inv}
	Let $q$ be a prime power and $\gcd\left(r,q^2-1\right)=1$. Let $t$ be a positive integer satisfying $rt\equiv1\pmod{\left(q^2-1\right)}$. Then the compositional inverse of the Dickson polynomial $D_r(x,a)$ over $\gf_q$ is $$D_r^{-1}(x,a)=D_t\left(x,a^r\right).$$
\end{Lemma}

Next, we give the other class of $2$-to-$1$ trinomials over $\gf_{2^n}$.

\begin{Th}
	\label{trinomials-2}
	Let $f(x)=x^{\frac{2^{n-1}+2^m-1}{3}}+x^{2^m}+\omega x\in\gf_{2^n}[x]$, where $n=2m$, $m$ is odd and $\omega\in\gf_{2^2}\backslash\gf_2$. Then $f(x)$ is $2$-to-$1$ over $\gf_{2^n}$.
\end{Th}

\begin{proof}
	According to the Euclidean algorithm, we can obtain the following equation
	$$(2^{2m}-1)\times\left(2^{m+1}+5\right)=\left(2^{2m-1}+2^m-1\right)\times\left(2^{m+2}+2\right)-3.$$
	Thus $\gcd\left(\frac{2^{n-1}+2^m-1}{3},2^{2m}-1\right)=1$. Let $g(x)=f\left(x^{2^{m+2}+2}\right)=x+x^{2^{m+1}+4}+\omega x^{2^{m+2}+2}$. Then we suffice to prove that $g(x)$  is $2$-to-$1$ over $\gf_{2^n}$. That is to say, we need to show that equation 
	\begin{equation}
	\label{trin2}
	g(x+a)+g(a)=0
	\end{equation}
    has exactly two solutions in $\gf_{2^n}$ for any $a\in\gf_{2^n}$. Let $y=x^{2^m}$ and $b=a^{2^m}$.  Then $g(x)=x+x^4y^2+\omega x^2y^4$ and $g(x+a) = x+a + (x+a)^4(y+b)^2+\omega (x+a)^2(y+b)^4$. 
      Furthermore, after simplifying, Eq. (\ref{trin2})  becomes
    \begin{equation}
    \label{trin2-1}
   x+y^2x^4+b^2x^4+\omega y^4x^2+\omega b^4x^2+\omega a^2y^4+a^4y^2=0.
    \end{equation}
  Computing $(\ref{trin2-1})+\omega\times(\ref{trin2-1})^{2^m}$ and simplifying, we obtain $y=\omega^2x$. Plugging it into Eq. (\ref{trin2-1}), we have 
  \begin{equation}
  \label{trin2-2}
  x^6+\omega\left(a^2+b^2\right)x^4+\omega^2\left(a^4+b^4\right)x^2+\omega x=0.
  \end{equation}
   It is clear that $x=0$ is a solution of Eq. (\ref{trin2-2}). Thus it suffices to show that equation 
   \begin{equation}
 \label{trin2-3}
 x^5+\omega\left(a^2+b^2\right)x^3+\omega^2\left(a^4+b^4\right)x+\omega =0
 \end{equation}  
 has exactly one solution in $\gf_{2^n}^{*}$ for any $a\in\gf_{2^n}$. 
 
 When $a\in\gf_{2^m}$, $b=a^{2^m}=a$ and Eq. (\ref{trin2-3}) becomes $x^5=\omega$, which has exactly one solution in $\gf_{2^n}^{*}$ since $\gcd\left(5,2^{2m}-1\right)=1$.
 
 When $a\in\gf_{2^n}\backslash\gf_{2^m}$, let $c=a^2+b^2$. Then $c\in\gf_{2^m}^{*}$. Let $F(x)=x^5+\omega c x^3+\omega^2c^2x+\omega$. Firstly, we claim that $F(x)$ can be decomposed into the product of a binomial and a trinomial. Assume that
 \begin{eqnarray*}
 F(x) &=&\left(x^2+A_1x+A_2\right)\left(x^3+A_1x^2+A_3x+A_4\right)\\
 &=& x^5+\left(A_3+A_1^2+A_2\right)x^3+\left(A_4+A_1A_3+A_1A_2\right)x^2+\left(A_1A_4+A_2A_3\right)x+A_2A_4,
 \end{eqnarray*}
 where $A_1,A_2,A_3,A_4\in\gf_{2^n}.$ After comparing the coefficients of the above equation, we have   
  \begin{subequations} \label{A1A2A3A4}
  	\renewcommand\theequation{\theparentequation.\arabic{equation}}     
  	\begin{empheq}[left={\empheqlbrace\,}]{align}
  	& ~ A_3+A_1^2+A_2 =  \omega c    \label{A1A2A3A4:1A} \\
  	 &~  A_4+A_1A_3+A_1A_2 =  0         \label{A1A2A3A4:1B} \\
  	 & ~ A_1A_4+A_2A_3 =  \omega^2c^2 \label{A1A2A3A4:1C} \\
  	& ~  A_2A_4 =  \omega \label{A1A2A3A4:1D}.
  	\end{empheq}
  \end{subequations}
From Eq. (\ref{A1A2A3A4:1A}), we have $A_2+A_3=A_1^2+\omega c$. Together with Eq. (\ref{A1A2A3A4:1B}) and Eq. (\ref{A1A2A3A4:1C}), we obtain $A_4=A_1(A_2+A_3)=A_1^3+\omega cA_1$ and $A_2A_3=A_1A_4+\omega c^2=A_1^4+\omega cA_1^2+\omega^2c^2.$ Thus $A_2, A_3$ are two solutions of the equation 
\begin{equation}
\label{A2A3}
x^2+\left(A_1^2+\omega c\right)x+A_1^4+\omega cA_1^2+\omega^2c^2=0.
\end{equation}
Let $x=\left(A_1^2+\omega c\right)y$. Then Eq. (\ref{A2A3}) becomes $$y^2+y+\frac{\omega cA_1^2+\omega^2c^2}{A_1^4+\omega^2 c^2}=0.$$
Since 
\begin{eqnarray*}
& & \frac{\omega cA_1^2+\omega^2c^2}{A_1^4+\omega^2 c^2} \\
&=& 1+\frac{\omega c A_1^2+\omega^2c^2+\omega^2c^2}{A_1^4+\omega^2 c^2} \\
&=& \omega+\omega^2+\frac{\omega c}{A_1^2+\omega c} +\frac{\omega^2 c^2}{A_1^4+\omega^2 c^2},
\end{eqnarray*}
we have $\left(A_2,A_3\right)=\left(\omega A_1^2+  c, \omega^2 A_1^2 + \omega^2 c\right)$ or $\left(A_2,A_3\right)=\left(\omega^2 A_1^2 + \omega^2 c, \omega A_1^2+  c \right)$. 

{\bfseries Case 1:} $\left(A_2,A_3\right) = \left(\omega A_1^2+  c, \omega^2 A_1^2 + \omega^2 c\right)$. Together with Eq. (\ref{A1A2A3A4:1D}), we know that $A_1$ satisfies
\begin{equation}
\label{A1}
x^5+cx^3+c^2x+1=0.
\end{equation}
 Let $x = \tilde{x}c^{\frac{1}{2}}$. Then Eq. (\ref{A1}) turns to 
\begin{equation}
\label{tilde(A1)}
\tilde{x}^5 + \tilde{x}^3 + \tilde{x} + c^{-\frac{5}{2}}=0.
\end{equation} 
What should be noticed is that $c=a^2+b^2\in\gf_{2^m}^{*}$. Let $D(x) = x^5+x^3+x = D_5(x,1)$.  Thus $D(x)$ permutes $\gf_{2^m}$. Furthermore, according to Lemma \ref{lem_Dick_inv}, we can compute that 
\begin{equation}
\label{A1_value}
A_1=c^{\frac{1}{2}}D_t\left(c^{-\frac{5}{2}},1\right),
\end{equation}
 where $5t\equiv1\pmod {\left(2^n-1\right)}$.  In the following, we claim that Eq. (\ref{A1}) does have exactly one solution, even if in $\gf_{2^n}$, i.e., $A_1\in\gf_{2^m}$. Since there exists a unique solution in $\gf_{2^m}$ for Eq. (\ref{A1}), we have the decomposition below:
\begin{eqnarray*}
& & x^5+cx^3+c^2x+1 \\
& = & \left(x+A_1\right)\left(x^4+A_1x^3+B_1x^2+B_2x+B_3\right) \\
& = & x^5+\left(B_1+A_1^2\right) x^3 +\left(B_2+A_1B_1\right) x^2 + \left(B_3+A_1B_2\right) x + A_1B_3,
\end{eqnarray*}
where $B_1,B_2,B_3\in\gf_{2^m}$. After comparing the coefficients of the above decomposition and simplifying, we get 
  \begin{subequations} \label{B1B2B3}
	\renewcommand\theequation{\theparentequation.\arabic{equation}}     
	\begin{empheq}[left={\empheqlbrace\,}]{align}
	B_1 = &~ A_1^2+c    \label{B1B2B3:1A} \\
	B_2 = &~ A_1^3+cA_1         \label{B1B2B3:1B} \\
	B_3 = &~ \frac{1}{A_1} = A_1^4+cA_1^2+c^2  \label{B1B2B3:1C}. 
	\end{empheq}
\end{subequations}

Let $G(x)=x^4+A_1x^3+B_1x^2+B_2x+B_3,$ where $B_1,B_2,B_3$ are defined by (\ref{B1B2B3}). 
 Next we prove that $G(x) = 0$ has no solution in $\gf_{2^n}$. Assume that there exists $\gamma\in\gf_{2^n}\backslash\gf_{2^m}$ such that $G(\gamma) =0$. Then it is trivial that $G\left(\gamma^{2^m}\right)=0$. Let $u=\gamma+\gamma^{2^m}$ and $v=\gamma^{2^m+1}$. Then $u,v\in\gf_{2^m}$ and $x^2+ux+v \mid G(x) $.  Furthermore, $G(x)$ can factor as a product of two quadratic irreducible polynomials over $\gf_{2^m}$, i.e., $G(x)=(2,2)$ over $\gf_{2^m}$ since it is trivial that $G(x)$ has no solution in $\gf_{2^m}$.  Thus it suffices to show that $G(x)\neq(2,2)$ over $\gf_{2^m}$, which needs the help of Lemma \ref{lem_deg_4}.

 Let $\alpha^2=\frac{B_2}{A_1}=A_1^2+c$. Consider 
\begin{equation}
\label{G(x)}
x^4G\left(\frac{1}{x}+\alpha\right) = D_0x^4+D_2x^2+A_1x+1,
\end{equation}
where 
\begin{eqnarray*}
D_0 &=& \alpha^4+A_1\alpha^3+B_1\alpha^2+B_2\alpha+B_3\\
& = & B_3 = A_1^4+cA_1^2+c^2
\end{eqnarray*}
and 
$$D_2=A_1\alpha+B_1=c^{\frac{1}{2}}A_1+c.$$
It is clear that $D_0=B_3\neq0$ since $B_3$ also equals $\frac{1}{A_1}$ from Eq. (\ref{B1B2B3:1C}). Let $\tilde{G}(x)=\frac{1}{D_0}x^4G\left(\frac{1}{x}+\alpha\right) = x^4+\frac{D_2}{D_0}x^2+\frac{A_1}{D_0}x+\frac{1}{D_0}$. Then we suffice to show that $\tilde{G}(x)\neq (2,2)$ over $\gf_{2^m}$. Consider 
\begin{equation}
\label{char-Gx}
H(y) = y^3+\frac{D_2}{D_0} y^2+\frac{A_1}{D_0}=0.
\end{equation}
Since
\begin{eqnarray*}
&   & \tr_{2^m}\left(\left.\frac{D_2^3}{D_0^3}\right/\frac{A_1^2}{D_0^2}\right) \\
& = & \tr_{2^m}\left(\frac{\left(c^{\frac{1}{2}}A_1+c\right)^3}{A_1^2\left(A_1^4+cA_1^2+c^2\right)}\right) \\
& = & \tr_{2^m}\left(\frac{c^3A_1^6+c^6+c^4A_1^4+c^5A_1^2}{A_1^{12}+c^2A_1^8+c^4A_1^4}\right)\\
& = & \tr_{2^m}\left(\frac{c^3}{A_1^6+cA_1^4+c^2A_1^2}+\frac{c^6}{A_1^{12}+c^2A_1^8+c^4A_1^4}\right)\\
& = & 0,
\end{eqnarray*}
$\tr_{2^m}\left(\left.\frac{D_2^3}{D_0^3}\right/\frac{A_1^2}{D_0^2}+1\right)=\tr_{2^m}(1)=1$. Furthermore, Eq. (\ref{char-Gx}) has exactly one solution in $\gf_{2^m}$, i.e., $H \neq (1,1,1)$. Moreover, $\tilde{G}\neq(2,2)$ according to Lemma \ref{lem_deg_4}.  

Up to now, we have proved that $F(x)=x^5+\omega c x^3+\omega^3 c^2 x+\omega$ can be decomposed into one binomial and a trinomial. That is to say, there exists $A_1\in\gf_{2^m}$ such that $$F(x) = \left(x^2+A_1x+A_2\right)\left(x^3+A_1x^2+A_3x+A_4\right), $$
where $A_2=\omega A_1^2+c$, $A_3=\omega^2A_1^2+\omega^2c$, $A_4=A_1^3+\omega c A_1$ and under the relationship above, $A_1$ is unique. On one hand, consider $F_1(x)=x^2+A_1x+A_2=0$ firstly. Since 
\begin{eqnarray*}
& & \tr_{2^n}\left(\frac{A_2}{A_1^2}\right) \\
&=& \tr_{2^n}\left(\frac{\omega A_1^2+c}{A_1^2}\right) \\
&=& \tr_{2^n}\left(\omega+\frac{c}{A_1^2}\right) \\
&=& 1,
\end{eqnarray*}
where the last equal sign holds thanks to $\tr_{2^n}\left(\omega\right)=m\equiv1\pmod2$ and $\frac{c}{A_1^2}\in\gf_{2^m}$, $F_1(x)=0$ has no solution in $\gf_{2^n}$ according to Lemma \ref{lem_deg_2}. On the other hand, consider $F_2(x)=x^3+A_1x^2+A_3x+A_4=0$. In order to show that $F(x)=F_1(x)F_2(x)=0$ has exactly one solution in $\gf_{2^n}$, it suffices to prove that 
\begin{equation}
\label{F2}
F_2(x)=x^3+A_1x^2+A_3x+A_4=0
\end{equation}
 has exactly one solution in $\gf_{2^n}$. Let $x=A_1y$. Then Eq. (\ref{F2}) turns to
 \begin{equation}
 \label{F2-y}
 y^3+\left(A_1^2+A_3\right)y+A_1A_3+A_4=0.
 \end{equation}
Since $A_3=\omega^2A_1^2+\omega^2c$ and $A_4=A_1^3+\omega c A_1$, we have $A_1^2+A_3=\omega A_1^2+\omega^2c\neq0$ and $A_1A_3+A_4 = \omega A_1^3+cA_1$. From Lemma \ref{lem_deg_3}, we suffice to claim that 
$$\tr_{2^n}\left(\frac{\left(A_1^2+A_3\right)^3}{\left(\omega A_1^3+cA_1\right)^2}+1\right)=\tr_{2^n}\left(\frac{A_1^6+\omega c A_1^4+\omega^2c^2A_1^2+c^3}{\omega^2 A_1^6 + c^2 A_1^2}\right)=1.$$
As a matter of fact, 
\begin{eqnarray*}
& &\tr_{2^n}\left(\frac{A_1^6+\omega c A_1^4+\omega^2c^2A_1^2+c^3}{\omega^2 A_1^6 + c^2 A_1^2}\right) \\
&=&\tr_{2^n}\left(\omega+\frac{\omega c A_1^4+c^2A_1^2+c^3}{\omega^2 A_1^6 + c^2 A_1^2}\right) \\
&=&1+\tr_{2^n}\left(\frac{\omega D^4 + D^2 + 1}{\omega^2 D^6+D^2}\right)\\
&=&1+\tr_{2^n}\left(\frac{D}{\omega D^3+D}+\frac{1}{\omega^2 D^6 + D^2}\right)\\
&=&1+\tr_{2^n}\left(\frac{D+1}{\omega D^3+D}\right)\\
&=&1+\tr_{2^n}\left(\frac{1}{D}+\frac{\omega D+1}{\omega D^2+1}\right)\\
&=&1+\tr_{2^n}\left(\frac{1}{D}\right)+\tr_{2^n}\left(\frac{D+\omega^2}{D^2+\omega^2}\right)\\
&=&1+\tr_{2^n}\left(\frac{1}{D}\right)+\tr_{2^n}\left(\frac{1}{D+\omega}+\frac{1}{D^2+\omega^2}\right)\\
&=&1,
\end{eqnarray*}
where $D = D_t\left(c^{-\frac{5}{2}},1\right) = c^{-\frac{1}{2}} A_1$ from (\ref{A1_value}) and the last equation holds since $D\in\gf_{2^m}$. Thus Eq. (\ref{F2}) has exactly one solution in $\gf_{2^n}$ and so does Eq. (\ref{trin2-3}). 

{\bfseries Case 2: $\left(A_2,A_3\right)=\left(\omega^2 A_1^2 + \omega^2 c, \omega A_1^2+  c \right)$}. The proof of this case is similar with Case 1 and we omit it here.  

To sum up, Eq. (\ref{trin2-3}) has exactly one solution in $\gf_{2^n}$. It follows that $f(x)$ is $2$-to-$1$ over $\gf_{2^n}$.
\end{proof}

%
\begin{Rem}
	\emph{	Note that the reason why there exist two cases in the proof of the above theorem. In fact, $F(x)=x^5+\omega cx^3+\omega^2 c^2 x^2 + \omega = \left(x+x_0\right) \tilde{F_1}(x)\tilde{F_2}(x), $ where $\tilde{F_1}(x), \tilde{F_2}(x)$ are both quadratic and irreducible over $\gf_{2^n}$ and $x_0$ is the unique solution in $\gf_{2^n}$ of $F(x)=0$. The $A_2$ in Case 1  and that in Case 2 are the coefficients of $x$ of $\tilde{F_1}$ and $\tilde{F_2}$, respectively.  
	}
\end{Rem}

\section{$2$-to-$1$ quadrinomials over $\gf_{2^n}$}
In this section, we consider $2$-to-$1$ quadrinomials of the form $x^k + x^{\ell} + x^d + x$ over $\gf_{{2^n}}$. In detail, we present ten classes of $2$-to-$1$ quadrinomials over $\gf_{{2^n}}$ with $n$ odd and two classes with $n\equiv0\pmod 3$.

\subsection{The infinite classes for $n$ odd}

In the section, we present ten classes of $2$-to-$1$ quadrinomials of the form $x^{k}+x^{\ell} + x^d + x$ over $\gf_{2^n}$ with $n$ odd. An important tool we will use is the resultant of two polynomials. 

\begin{Th}
	\label{n-odd-1}
	Let $n=2m+1$ and $f(x)=x^{2^{m+1}+2}+x^{2^{m+1}}+x^2+x \in\gf_{2^n}[x]$. Then $f(x)$ is $2$-to-$1$ over $\gf_{2^n}$.
\end{Th}

\begin{proof}
		Let $y=x^{2^{m+1}}$ and $b=a^{2^{m+1}}$ for any $a\in\gf_{2^n}$. Then  $f(x)=x^2y+y+x^2+x$ and $f(x+a)+f(a)= yx^2+\left(a^2+1\right)y+(b+1)x^2+x$ after computing and simplifying and we suffice to show that for any $a\in\gf_{2^n}$, $$f(x+a)+f(a)=0,$$
	i.e.,
	\begin{equation}
	\label{four1}
	 yx^2+\left(a^2+1\right)y+(b+1)x^2+x = 0
	\end{equation}
	has exactly two solutions in $\gf_{2^n}$. It is clear that when $a=0$ or $ 1$, Eq. (\ref{four1}) has exactly two solutions $x=0,1$. Therefore, in the rest of this proof, we assume that $a\neq0,1$ and $x\neq0$. That is to say, it suffices to prove that Eq. (\ref{four1}) has exactly one nonzero solution in $\gf_{2^n}$ for all $a\in\gf_{2^n}\backslash\{0, 1\}$.
	
	Raising Eq. (\ref{four1}) into its $2^{m+1}$-th power and simplifying it by $y^{2^{m+1}}=x^2$ and $b^{2^{m+1}}=a^2$, we have 
	\begin{equation}
\label{four1_1}
x^2y^2+\left(b^2+1\right)x^2+\left(a^2+1\right)y^2+y=0.
\end{equation}	
	Let $$F(x,y) =  yx^2+\left(a^2+1\right)y+(b+1)x^2+y$$ and
$$G(x,y) = x^2y^2+\left(b^2+1\right)x^2+\left(a^2+1\right)y^2+y.$$
After computing, we have
$$ \mathrm{Res}(F,G,y) = x(x+a+1)\left( \left( a^2b^2+a^2+b^2+b+1 \right) x+1 \right).$$
	Therefore, from Eq. (\ref{four1}) and Eq. (\ref{four1_1}), $x=a+1$ or $ \left( a^2b^2+a^2+b^2+b+1 \right) x+1 =0$. In the following, we claim $x\neq a+1$. Suppose the contrary, that $x=a+1$, then $y=b+1$ and from Eq. (\ref{four1}), $(b+1)\left(a^2+1\right)+(a+1)=0$, i.e., $ab+a+b=0$. After raising $ab+a+b=0$ into its $2^{m+1}$-th power, we have $ba^2+b+a^2=0.$ Together with $ab+a+b=0$, we get $a=0$, which leads to a contradiction. 
Hence, it suffices to show $a^2b^2+a^2+b^2+b+1\neq0$ for $a\in\gf_{{2^n}}$. Assume that there exists some $a$ such that 
\begin{equation}
\label{four1_3}
a^2b^2+a^2+b^2+b+1=0.
\end{equation}
Computing $(\ref{four1_3})^{2^{m+1}}+(\ref{four1_3})\times a^2$, we have $a^2\left(b^2+b\right)=b^2+1$, i.e., $a^2b=b+1$ since $b\neq 1$. Together with Eq. (\ref{four1_3}), we get $a^2=1$, leading $a=1$, which is also conflict. 
Thus $x=\frac{1}{a^2b^2+a^2+b^2+b+1}$, which completes the  proof.
\end{proof}

\begin{Th}
	\label{n-odd-2}
	Let $n=2m+1$ and $f(x)=x^{2^{m+1}+2}+x^{2^{m+1}+1}+x^2+x \in\gf_{2^n}[x]$. Then $f(x)$ is $2$-to-$1$ over $\gf_{2^n}$.
\end{Th}

\begin{proof}
	Let  $y=x^{2^{m+1}}$ and $b=a^{2^{m+1}}$ for any $a\in\gf_{2^n}$. Then $f(x)=x^2y+xy+x^2+x$ and $f(x+a)+f(a)=x^2y+xy+(b+1)x^2+\left(a^2+a\right)y+(b+1)x$ after computing and simplifying. According to the definition, it suffices to show that for any $a\in\gf_{2^n}$, $$f(x+a)+f(a)=0,$$
	i.e.,
	\begin{equation}
	\label{four2}
	x^2y+xy+(b+1)x^2+\left(a^2+a\right)y+(b+1)x = 0
	\end{equation}
	has exactly two solutions in $\gf_{2^n}$.  It is clear that $x=0$ is a solution of Eq. (\ref{four2}) and when $a=1$, Eq. (\ref{four2}) has exactly two solutions $x=0,1$ in $\gf_{{2^n}}$. In addition, when $a=0$, Eq. (\ref{four2}) becomes $\left(x^2+x\right)(y+1)=0$, meaning $x=0$ or $1$. In the following, we assume $a\in\gf_{2^n}\backslash\{0,1\}$ and consider the nonzero solution of Eq. (\ref{four2}) in $\gf_{2^n}$. 
	
	Raising Eq. (\ref{four2}) into its $2^{m+1}$-th power and simplifying it by $y^{2^{m+1}} = x^2$ and $b^{2^{m+1}} = a^2$, we obtain 
	\begin{equation}
	\label{four2_1}
	x^2y^2+x^2y+\left(a^2+1\right)y^2+\left(b^2+b\right)x^2+\left(a^2+1\right)y = 0.
	\end{equation}
	Let $$F(x,y) =  	x^2y+xy+(b+1)x^2+\left(a^2+a\right)y+(b+1)x = 0 $$ and
$$G(x,y) = 	x^2y^2+x^2y+\left(a^2+1\right)y^2+\left(b^2+b\right)x^2+\left(a^2+1\right)y.$$
After computing, we have
$$ \mathrm{Res}(F,G,y) =  (a+1)(b+1)x(x+a+1)\left( (ab+a+b)x+a\right).$$

Similarly, it is easy to check that $x=\frac{a}{ab+a+b}$ is the only solution of Eq. (\ref{four2}) in $\gf_{2^n}^{*}$.	This completes the  proof.
\end{proof}

\begin{Th}
	\label{n-odd-3}
	Let $n=2m+1$ and $f(x)=x^{2^{m+2}+4}+x^{2^{m+1}+2}+x^2+x \in\gf_{2^n}[x]$. Then $f(x)$ is $2$-to-$1$ over $\gf_{2^n}$.
\end{Th}

\begin{proof}
		Let  $y=x^{2^{m+1}}$ and $b=a^{2^{m+1}}$ for any $a\in\gf_{2^n}$. Then $f(x) = x^4y^2+x^2y+x^2+x$ and $f(x+a)+f(a) = x^4y^2+b^2x^4+a^4y^2+x^2y+bx^2+a^2y+x^2+x$ after simplifying. We suffice to show that for any $a\in\gf_{2^n}$, 
		$$f(x+a)+f(a)=0,$$
		i.e.,
		\begin{equation}
		\label{four3}
		x^4y^2+b^2x^4+a^4y^2+x^2y+bx^2+a^2y+x^2+x = 0
		\end{equation}
		has exactly two solutions in $\gf_{2^n}$. When $a=0$, Eq. (\ref{four3}) becomes $x^4y^2+x^2y+x^2+x = x(xy+1)\left(x^2y+x+1\right) =0$ and it is easy to check that $x=0, 1$  are all solutions of Eq. (\ref{four3}) in $\gf_{2^n}$. Similarly, we can show that when $a=1$, Eq. (\ref{four3}) also has exactly two solutions $x=0,1$ in $\gf_{2^n}$. In the following, we assume $a\in\gf_{2^n}\backslash\{0,1\}$.
		
		Raising Eq. (\ref{four3}) into its $2^{m+1}$-th power, we obtain 
		\begin{equation}
        \label{four3_1}
        x^4y^4+a^4y^4+b^4x^4+y^2x^2+a^2y^2+b^2x^2+y^2+y = 0.
        \end{equation}		
        Let $$F(x,y) = 	x^4y^2+b^2x^4+a^4y^2+x^2y+bx^2+a^2y+x^2+x$$ and 
        $$G(x,y) = 
        x^4y^4+a^4y^4+b^4x^4+y^2x^2+a^2y^2+b^2x^2+y^2+y.$$
        Then due to MAGMA, we obtain
        $$ \mathrm{Res}(F,G,y) = x^2(x+a)^8\left(x+\frac{1}{b}\right)^2\left(a^2bx+x+a^2\right)^2\left(a^2b^2x+bx+x+1\right)^2.$$ 
       Since $\gcd\left(2^{m+1}+2, 2^n-1\right)=1$, $a^2b+1\neq0$ for any $a\in\gf_{2^n}\backslash\{0,1\}$. In addition, we claim that $a^2b^2+b+1\neq0$. Assume $a^2b^2+b+1=0$, raising it into its $2^{m+1}$-th power, we get $a^4b^2+a^2+1=0$, i.e., $a^2b+a+1=0$. Together with $a^2b^2+b+1=0$, we have $ab+1=0$, which leads to a contradiction since $a^2b^2+b+1=0$. Thus the solutions of Eq. (\ref{four3}) may have $x=0, a, \frac{1}{b}, \frac{a^2}{a^2b+1}, \frac{1}{a^2b^2+b+1}$. After checking carefully, we find that Eq. (\ref{four3}) has exactly two solutions $x=0$ and $\frac{1}{a^2b^2+b+1}$, which completes the  proof.
\end{proof}

\begin{Th}
	\label{n-odd-4}
	Let $n=2m+1$ and $f(x)=x^{2^{n}-2^{m+1}+2}+x^{2^{m+1}}+x^2+x \in\gf_{2^n}[x]$. Then $f(x)$ is $2$-to-$1$ over $\gf_{2^n}$.
\end{Th}

\begin{proof}
		Let  $y=x^{2^{m+1}}$ and $b=a^{2^{m+1}}$ for any $a\in\gf_{2^n}$. Then $f(x) = \frac{x^3}{y} + y + x^2 + x$ and $f(x+a)+f(a) = \frac{(x+a)^3}{y+b}+\frac{a^3}{b} + y + x^2+x,$ where $\frac{1}{0}$ equals $0$. We suffice to show that for any $a\in\gf_{2^n}$, 
		$$f(x+a)+f(a)=0,$$
		i.e.,
		\begin{equation}
		\label{four4}
		\frac{(x+a)^3}{y+b}+\frac{a^3}{b} + y + x^2+x = 0
		\end{equation}
		has exactly two solutions in $\gf_{2^n}$. When $a=0$, Eq. (\ref{four4}) turns to $\frac{x^3}{y}+y+x^2+x =\frac{(x+y)\left(x^2+y\right)}{y} =0$, which has exactly two solutions $x=0,1$ in $\gf_{2^n}$ since $\gcd\left(2^{m+1}-2,2^n-1\right)=1$ and $\gcd\left(2^{m+1}-1,2^n-1\right)=1$.  On the other hand, if $a\neq0$ and $x=a$ is a solution of  Eq. (\ref{four4}),  then $\frac{a^3}{b}+b+a^2+a=0$, which means $a=1$. In addition, when $a=1$, it is easy to check that Eq. (\ref{four4}) has also exactly two solutions $x=0,1$ in $\gf_{2^n}$. Thus in the following, we assume $a\in\gf_{2^n}\backslash\{0,1\}$, which means $x=a$ is not a solution of Eq. (\ref{four4}).  After multiplying $b(y+b)$ to Eq. (\ref{four4}) and simplifying, we have 
		\begin{equation}
		\label{four4_1}
		by^2+\left(bx^2+bx+a^3+b^2\right)y+bx^3+\left(ab+b^2\right)x^2+\left(a^2b+b^2\right)x=0.
		\end{equation}  
		Raising Eq. (\ref{four4_1}) into its $2^{m+1}$-th power, we get 
		\begin{equation}
		\label{four4_2}
		a^2x^4+\left(a^4+b^3+a^2y^2+a^2y\right)x^2+a^2y^3+\left(a^2b+a^4\right)y^2+\left(b^2a^2+a^4\right)y=0.
		\end{equation} 
		Let $$F(x,y) = 	by^2+\left(bx^2+bx+a^3+b^2\right)y+bx^3+\left(ab+b^2\right)x^2+\left(a^2b+b^2\right)x$$ and
		$$G(x,y) = 	a^2x^4+\left(a^4+b^3+a^2y^2+a^2y\right)x^2+a^2y^3+\left(a^2b+a^4\right)y^2+\left(b^2a^2+a^4\right)y.$$
		Then thanks to MAGMA, we get
		$$ \mathrm{Res}(F,G,y) = (a+b)^3\left(a^2+b\right)^2x(x+a)^2\left(bx+a^2\right)(ax+b)\left(abx+a^3+ab+b^2\right).$$
		Thus the common solution of Eqs. (\ref{four4_1}) and (\ref{four4_2}) may be $x=0, a, \frac{a^2}{b}, \frac{b}{a}, \frac{a^3+ab+b^2}{ab}$. After testing carefully, Eq. (\ref{four4}) has exactly two solutions  $x=0$ and $\frac{a^3+ab+b^2}{ab}$ in $\gf_{2^n}$.
\end{proof}

\begin{Th}
	\label{n-odd-5}
	Let $n=2m+1$ and $f(x)=x^{2^n-2}+x^{2^n-2^{m+1}}+x^{2^n-2^{m+1}-2}+x$. Then $f(x)$ is $2$-to-$1$ over $\gf_{{2^n}}$.
\end{Th}

\begin{proof}
	Let  $y=x^{2^{m+1}}$ and $b=a^{2^{m+1}}$ for any $a\in\gf_{2^n}$. Then $f(x) = \frac{1}{x} + \frac{x}{y} + \frac{1}{xy} + x$ and $f(x+a)+f(a) = \frac{1}{x+a}+\frac{x+a}{y+b}+\frac{1}{(x+a)(y+b)}+x+\frac{1}{a}+\frac{1}{ab}+\frac{a}{b},$ where $\frac{1}{0}$ equals $0$. We suffice to show that for any $a\in\gf_{2^n}$, 
	$$f(x+a)+f(a)=0,$$
	i.e.,
	\begin{equation}
	\label{four5}
\frac{1}{x+a}+\frac{x+a}{y+b}+\frac{1}{(x+a)(y+b)}+x= \frac{a^2+b+1}{ab}
	\end{equation}
	has exactly two solutions in $\gf_{2^n}$. When $a=0$, Eq. (\ref{four5}) turns to $ \frac{1}{x} + \frac{x}{y} + \frac{1}{xy} + x =\frac{x^2y+x^2+y+1}{xy} =0$, which has exactly two solutions $x=0,1$ in $\gf_{2^n}$.  On the other hand, if $a\neq0$ and $x=a$ is a solution of  Eq. (\ref{four5}),  then $\frac{1}{a} + \frac{a}{b} + \frac{1}{ab} + a =0$, which means $a=1$. In addition, when $a=1$, it is easy to check that Eq. (\ref{four5}) has also exactly two solutions $x=0,1$ in $\gf_{2^n}$. Thus in the following, we assume $a\in\gf_{2^n}\backslash\{0,1\}$, which means $x=a$ is not a solution of Eq. (\ref{four5}).
	After multiplying $ab(x+a)(y+b)$ to Eq. (\ref{four5}) and simplifying, we have 
	\begin{equation}
	\label{four5_1}
abx^2y + (ab^2 + ab) x^2  + (a^2b+a^2+b+1)  xy + (a^2b^2+ a^2b + b^2 +b) x + (a^3+a)  y=0.
	\end{equation}  
		Raising Eq. (\ref{four5_1}) into its $2^{m+1}$-th power, we obtain 
		\begin{equation}
	\label{four5_2}
a^2bx^2y^2+( a^4b+a^2b )y^2 +(a^2b^2+b^2+a^2+1)x^2y + (a^4b^2+a^2b^2 +a^4+a^2)y+(b^3+b)x^2 =0.
	\end{equation} 	
	Let $$F(x,y) = 		abx^2y + (ab^2 + ab) x^2  + (a^2b+a^2+b+1)  xy + (a^2b^2+ a^2b + b^2 +b) x + (a^3+a)  y$$ and
	$$G(x,y) = 	a^2bx^2y^2+( a^4b+a^2b )y^2 +(a^2b^2+b^2+a^2+1)x^2y + (a^4b^2+a^2b^2 +a^4+a^2)y+(b^3+b)x^2.$$
	Then the resultant of $F$ and $G$ aiming to $y$ is
	$$ \mathrm{Res}(F,G,y) = ab(a+1)^2(b+1)^2x(x+a)^2(x+a+1)^2(abx+a^2b+a^2+b+1).$$ 
Clearly,  $a^2b+a^2+b+1=(a^2+1)(b+1)\neq0$ for any $a\in \gf_{2^n}\backslash\{0,1\}$. After checking carefully, we have Eq. (\ref{four5}) has exactly two solutions $x=0$ and $x= \frac{a^2b+a^2+b+1}{ab}$ in $\gf_{{2^n}}$.
\end{proof}

\begin{Th}
	\label{n-odd-6}
	Let $n=2m+1$ and $f(x) = x^{2^n-2}+x^{2^n-2^{m+1}}+x^{2^{m+1}-1}+x$. Then $f(x)$ is $2$-to-$1$ over $\gf_{{2^n}}$.
\end{Th}

\begin{proof}
	Let  $y=x^{2^{m+1}}$ and $b=a^{2^{m+1}}$ for any $a\in\gf_{2^n}$. Then $f(x) = \frac{1}{x} + \frac{x}{y} + \frac{y}{x} + x$ and $f(x+a)+f(a) = \frac{1}{x+a}+\frac{x+a}{y+b}+\frac{y+b}{x+a}+x+\frac{1}{a}+\frac{b}{a}+\frac{a}{b},$ where $\frac{1}{0}$ equals $0$. It suffices to show that for any $a\in\gf_{2^n}$, 
	$$f(x+a)+f(a)=0,$$
	i.e.,
	\begin{equation}
	\label{n-odd-8-eq-1}
	\frac{1}{x+a}+\frac{x+a}{y+b}+\frac{y+b}{x+a}+x= \frac{a^2+b^2+b}{ab}
	\end{equation}
	has exactly two solutions in $\gf_{2^n}$. When $a=0$, Eq. (\ref{n-odd-8-eq-1}) turns to $ \frac{1}{x} + \frac{x}{y} + \frac{y}{x} + x =\frac{x^2y+x^2+y^2+y}{xy} =0$, which has exactly two solutions $x=0,1$ in $\gf_{2^n}$.  On the other hand, if $a\neq0$ and $x=a$ is a solution of  Eq. (\ref{n-odd-8-eq-1}),  then $\frac{1}{a} + \frac{a}{b} + \frac{b}{a} + a =0$, which means $a=1$. In addition, when $a=1$, it is easy to check that Eq. (\ref{n-odd-8-eq-1}) has also exactly two solutions $x=0,1$ in $\gf_{2^n}$. Thus in the following, we assume $a\in\gf_{2^n}\backslash\{0,1\}$, which means $x=a$ is not a solution of Eq. (\ref{n-odd-8-eq-1}).
	After multiplying $ab(x+a)(y+b)$ to Eq. (\ref{n-odd-8-eq-1}) and simplifying, we have 
	\begin{equation}
	\label{n-odd-8-eq-2}
	abx^2y+(ab^2+ab)x^2+(a^2b+a^2+b^2+b)xy+(a^2b^2+a^2b+b^3+b^2)x+aby^2+(a^3+ab^2)y=0.
	\end{equation}
	Raising Eq. (\ref{n-odd-8-eq-2}) into its $2^{m+1}$-th power, we obtain 
	\begin{equation}
	\label{n-odd-8-eq-3}
	a^2bx^2y^2+(a^4b+a^2b)y^2+(a^2b^2+b^2+a^4+a^2)x^2y+(a^4b^2+a^2b^2+a^6+a^4)y+a^2bx^4+(b^3+a^4b)x^2=0.
	\end{equation}		
	Let $$F(x,y) = 		abx^2y+(ab^2+ab)x^2+(a^2b+a^2+b^2+b)xy+(a^2b^2+a^2b+b^3+b^2)x+aby^2+(a^3+ab^2)y$$ and
	$$G(x,y) = 	a^2bx^2y^2+(a^4b+a^2b)y^2+(a^2b^2+b^2+a^4+a^2)x^2y+(a^4b^2+a^2b^2+a^6+a^4)y+a^2bx^4+(b^3+a^4b)x^2.$$
	Then the resultant of $F$ and $G$ aiming to $y$ is	
	$$ \mathrm{Res}(F,G,y) = a^2b^2x^2(x+a)^2(bx+ab+a)^2(ax+a^2+1)^2(ax+a^2+b)^2.$$
	Then the solutions of Eq. (\ref{n-odd-8-eq-1}) may be $x=a+\frac{a}{b}, a+\frac{1}{a}$ and $a+\frac{b}{a}$. It is easy to check that $x=a+\frac{1}{a}$ is the unique solution of Eq. (\ref{n-odd-8-eq-1}) in $\gf_{{2^n}}^{*}$.
\end{proof}

\begin{Th}
	\label{n-odd-7}
	Let $n=2m+1$ and $f(x) = x^{2^n-2}+x^{2^{n-1}+1}+x^{2^{n-1}-2}+x$. Then $f(x)$ is $2$-to-$1$ over $\gf_{{2^n}}$.
\end{Th}
\begin{proof}
	Let  $g(x)=f(x)^2=\frac{1}{x^2}+\frac{1}{x^3}+x^3+x^2$, where $\frac{1}{0}=0$. Then it suffices to show that for any $a\in\gf_{{2^n}}$, $g(x+a)+g(x)=0$, i.e., 
	\begin{equation}
	\label{n-odd-7-eq-1}
	\frac{1}{(x+a)^2}+\frac{1}{(x+a)^3}+(x+a)^3+x^2+\frac{a^6+a+1}{a^3}=0
	\end{equation}
	has exactly two solutions in $\gf_{{2^n}}$. When $a=0$, Eq. (\ref{n-odd-9-eq-1}) becomes $\frac{1}{x^2}+\frac{1}{x^3}+x^3+x^2=\frac{x^6+x^5+x+1}{x^3}=0,$ which has exactly two solutions $x=0,1$ in $\gf_{{2^n}}$. If $a\neq0$ and $x=a$ is a solution of  Eq. (\ref{n-odd-7-eq-1}),  then $\frac{1}{a^2}+\frac{1}{a^3}+a^3+a^2=0$ and thus $a=1$.  In addition, when $a=1$, it is easy to check that Eq. (\ref{n-odd-7-eq-1}) has also exactly two solutions $x=0,1$ in $\gf_{2^n}$. Thus in the following, we assume $a\in\gf_{2^n}\backslash\{0,1\}$, which means that $x=a$ is not a solution of Eq. (\ref{n-odd-7-eq-1}). 	After multiplying $a^3(x+a)^3$ to Eq. (\ref{n-odd-7-eq-1}) and simplifying, we have 
	\begin{equation}
	\label{n-odd-7-eq-2}
	x(ax+a^2+1)\left( a^2x^4+(a^3+a^2+a)x^3+(a+1)x^2+(a^5+a^4+a^2+a)x+a^6+a^4+a^2  \right) = 0.
	\end{equation}
	Therefore, $x=0, a+\frac{1}{a}$ or 
	\begin{equation}
	\label{n-odd-7-eq-3}
	a^2x^4+(a^3+a^2+a)x^3+(a+1)x^2+(a^5+a^4+a^2+a)x+a^6+a^4+a^2  =0.
	\end{equation}
	In the following, we show that Eq. (\ref{n-odd-7-eq-3}) has no solutions in $\gf_{{2^n}}$ for any $a\in\gf_{2^n}\backslash\{0,1\}$. Let $F(x)= 	a^2x^4+(a^3+a^2+a)x^3+(a+1)x^2+(a^5+a^4+a^2+a)x+a^6+a^4+a^2.$ Let $A=a^4+a^3+a^2+a+1$. Then after simplifying, we have 
	$$G(x)=\frac{1}{A}x^4F\left(\frac{1}{x}+a+1\right) = x^4+\frac{a^4+1}{A}x^2+\frac{a^3+a^2+a}{A}x+\frac{a^2}{A}. $$
	Next we consider the factorization of $G(x)$ using Lemma \ref{lem_deg_4}. Here, $$G_1(y) = y^3+\frac{a^4+1}{A}y+\frac{a^3+a^2+a}{A} = (y+1)\left( y^2+y+\frac{a^3+a^2+a}{A} \right).$$

 (i) When $\tr_{2^{n}}\left( \frac{a^3+a^2+a}{A}  \right) = 1$,	 $G_1(y)=0$ has exactly one solution $r_1=1$. Moreover, in Lemma \ref{lem_deg_4}, $\omega_1 = \frac{a^2}{A} \cdot \frac{A^2}{a^6+a^4+a^2} = \frac{A}{a^4+a^2+1}$ and 
	\begin{eqnarray*}
	\tr_{2^{n}}\left(\omega_1\right) &=& \tr_{2^{n}}\left( \frac{A}{a^4+a^2+1} \right)\\
	&=& \tr_{2^{n}}\left( 1+\frac{a^3+a}{a^4+a^2+1} \right) \\
	&=& \tr_{2^{n}}\left(\frac{a}{a^2+a+1}+\frac{a^2}{a^4+a^2+1}\right)+1\\
	&=& 1.
	\end{eqnarray*}
Therefore, according to Lemma \ref{lem_deg_4}, we know $G=(4)$ and $G(x)=0$, i.e., Eq. (\ref{n-odd-7-eq-3}) has no solutions in $\gf_{{2^n}}$.
	
(ii)	 When  $\tr_{2^{n}}\left( \frac{a^3+a^2+a}{A}  \right) = 0$, $G_1(y)=0$  has three solutions $r_1=1, r_2, r_3$ in $\gf_{{2^n}}$, where $r_2+r_3=1$. Moreover, in Lemma \ref{lem_deg_4}, $\omega_1 = \frac{A}{a^4+a^2+1}$, $\omega_2 = \frac{A}{a^4+a^2+1}r_2^2$ and $\omega_3=\frac{A}{a^4+a^2+1} r_3^2$. From the discussion of the above case (i), we know $\tr_{2^{n}}\left(\omega_1\right)=1$. In addition, we have $\omega_2+\omega_3=\omega_1$ and thus $\tr_{2^{n}}\left(\omega_2+\omega_3\right)=1$. Namely, there exists one element in $\omega_i$ for $i=1,2,3$ such that $\tr_{2^{n}}\left(\omega_i\right)=0$ and according to Lemma \ref{lem_deg_4}, we have $G = (2,2)$ and $G(x)=0$, i.e., Eq. (\ref{n-odd-7-eq-3}) has no solutions in $\gf_{{2^n}}$.

All in all, Eq. (\ref{n-odd-7-eq-1}) has exactly two solutions in $\gf_{{2^n}}$.
\end{proof}

\begin{Th}
	\label{n-odd-8}
	Let $n=2m+1$ and $f(x) = x^{2^n-2}+x^{2^n-4}+x^3+x$. Then $f(x)$ is $2$-to-$1$ over $\gf_{2^{n}}$.
\end{Th}

\begin{proof}
 Clearly, $f(x)=\frac{1}{x}+\frac{1}{x^3}+x^3+x$, where $\frac{1}{0}=0$. To prove the theorem, it suffices to show that for any $a\in\gf_{{2^n}}$, $f(x+a)+f(x)=0$, i.e., 
 \begin{equation}
 \label{n-odd-9-eq-1}
 \frac{1}{x+a}+\frac{1}{(x+a)^3}+(x+a)^3+x+\frac{a^6+a^2+1}{a^3}=0
 \end{equation}
  has exactly two solutions in $\gf_{{2^n}}$. When $a=0$, Eq. (\ref{n-odd-9-eq-1}) becomes $\frac{1}{x}+\frac{1}{x^3}+x^3+x=0,$ which has exactly two solutions $x=0,1$ in $\gf_{{2^n}}$. If $a\neq0$ and $x=a$ is a solution of  Eq. (\ref{n-odd-9-eq-1}),  then $\frac{1}{a}+\frac{1}{a^3}+a^3+a=0$ and thus $a=1$.  In addition, when $a=1$, it is easy to check that Eq. (\ref{n-odd-9-eq-1}) has also exactly two solutions $x=0,1$ in $\gf_{2^n}$. Thus in the following, we assume $a\in\gf_{2^n}\backslash\{0,1\}$, which means $x=a$ is not a solution of Eq. (\ref{n-odd-9-eq-1}). 	After multiplying $a^3(x+a)^3$ to Eq. (\ref{n-odd-9-eq-1}) and simplifying, we have 
  \begin{equation}
  \label{n-odd-9-eq-2}
  x(ax+a^2+1)\left( a^2x^4+(a^3+a)x^3+(a^2+1)x^2+(a^5+a)x+a^6+a^2  \right) = 0.
  \end{equation}
  Therefore, $x=0, a+\frac{1}{a}$ or 
  \begin{equation}
  \label{n-odd-9-eq-3}
  a^2x^4+(a^3+a)x^3+(a^2+1)x^2+(a^5+a)x+a^6+a^2 =0.
  \end{equation}
  In the following, we show that Eq. (\ref{n-odd-9-eq-3}) has no solutions in $\gf_{{2^n}}$ for any $a\in\gf_{2^n}\backslash\{0,1\}$. Let $F(x)= a^2x^4+(a^3+a)x^3+(a^2+1)x^2+(a^5+a)x+a^6+a^2.$ After simplifying, we have 
  $$G(x)=\frac{1}{a^4+1}x^4F\left( \frac{1}{x}+a+1 \right) = x^4+\frac{a^2+a+1}{a^2+1}x^2+\frac{a}{a^2+1}x+\frac{a^2}{a^4+1}.$$
  Let $G_1(y)=y^3+ \frac{a^2+a+1}{a^2+1}y+\frac{a}{a^2+1}=(y+1)\left(y+\frac{1}{a+1}\right)\left(y+\frac{a}{a+1}\right)$. Then the solutions of $G_1(y)=0$ are $r_1=1, r_2= \frac{1}{a+1}$ and $r_3=\frac{a}{a+1}$, respectively. Moreover, in Lemma \ref{lem_deg_4}, $\omega_1= 1$, $\omega_2=\frac{1}{a^2+1}$, $\omega_3=\frac{a^2}{a^2+1}=1+\omega_2$. Since $n$ is odd, we have $\tr_{2^n}(\omega_1)=1$ and $\tr_{2^n}(\omega_2+\omega_3)=1$. Therefore, there is only one element in $\omega_i$ for $i=1,2,3$ such that $\tr_{2^n}(\omega_i) =0$ and by Lemma \ref{lem_deg_4}, we know that $G(x)=(2,2)$. Namely, $G(x)=0$ has no solutions in $\gf_{{2^n}}$ for any $a\in\gf_{2^n}\backslash\{0,1\}$ and neither does Eq. (\ref{n-odd-9-eq-3}). Therefore, Eq. (\ref{n-odd-9-eq-1}) has exactly two solutions $x=0, a+\frac{1}{a}$ for any $a\in\gf_{{2^n}}$.  
\end{proof}

\begin{Th}
	\label{n-odd-9}
	Let $n=2m+1$ and $f(x)=x^6+x^4+x^3+x$. Then $f(x)$  is $2$-to-$1$ over $\gf_{{2^n}}$.
\end{Th}

\begin{proof}
	It suffices to show that for any $a\in\gf_{{2^n}}$, $f(x+a)+f(a)=x^6+(a^2+1)x^4+x^3+(a^4+a)x^2+(a^2+1)x=0$ has exactly two solutions in $\gf_{{2^n}}$. Namely, we need to prove that
	\begin{equation}
	\label{n-odd-10-eq-1}
	x^5+(a^2+1)x^3+x^2+(a^4+a)x+a^2+1=0
	\end{equation}
	has exactly one solution in $\gf_{{2^n}}^{*}$. When $a=1$, Eq. (\ref{n-odd-10-eq-1}) becomes $x^5+x^2=0$, which means that Eq. (\ref{n-odd-10-eq-1}) has a unique solution $x=1$ in $\gf_{{2^n}}^{*}$ due to $\gcd(3,2^n-1)=1$. Moreover, it is easy to check that Eq. (\ref{n-odd-10-eq-1}) is actually
	$$\left( x^2+(a+1)x+a^2+1 \right) \left( x^3+(a+1)x^2+(a^2+1)x+1 \right) = 0. $$
	Therefore, we have $x^2+(a+1)x+a^2+1=0$, which has no solutions in $\gf_{{2^n}}$ since $\tr_{2^n}(1)=1$, or $ x^3+(a+1)x^2+(a^2+1)x+1=0$. For the latter equation, we have $(x+(a+1))^3=(a+1)^3+1=a^3+a^2+a$ and thus Eq. (\ref{n-odd-10-eq-1}) has exactly one solution in $\gf_{{2^n}}^{*}$ since $\gcd\left(3,2^n-1\right)=1$, which completes the  proof.
\end{proof}

\begin{Th}
	\label{n-odd-10}
	Let $n=2m+1$ and $f(x)=x^6+x^5+x^3+x$. Then $f(x)$ is $2$-to-$1$ over $\gf_{{2^n}}$. 
\end{Th}

\begin{proof}
	We suffice to prove that for $a\in\gf_{{2^n}}$, $f(x+a)+f(a)=x(x^5+x^4+(a^2+a)x^3+x^2+(a^4+a)x+a^4+a^2+1)=0$ has exactly two solutions in $\gf_{{2^n}}$. Namely, we need to prove that
	\begin{equation}
	\label{n-odd-12-eq-1}
	x^5+x^4+(a^2+a)x^3+x^2+(a^4+a)x+a^4+a^2+1=0
	\end{equation} 
	has exactly one solution in $\gf_{{2^n}}^{*}$. Using $x+1$ to substitute $x$ in Eq. (\ref{n-odd-12-eq-1}) and simplifying, Eq. (\ref{n-odd-12-eq-1}) becomes
	\begin{equation}
	\label{n-odd-12-eq-2}
	x^5+(a^2+a)x^3+(a^2+a+1)x^2+(a^4+a^2+1)x=0
	\end{equation} 	
	and it suffices to show that Eq. (\ref{n-odd-12-eq-2}) has exactly one solution in $\gf_{{2^n}}\backslash\{1\}$. Clearly, $x=0$ is a solution of Eq. (\ref{n-odd-12-eq-2}). In the following, we show that $G(x)=x^4+(a^2+a)x^2+(a^2+a+1)x+a^4+a^2+1$ is irreducible and thus $G(x)=0$ has no solutions in $\gf_{{2^n}}$. Let $G_1(y)=y^3+(a^2+a)y+a^2+a+1=(y+1)(y^2+y+a^2+a+1)$. Since $\tr_{2^n}(a^2+a+1)=1$, $G_1(y)=0$ has exactly one solution $x_1=1$ in $\gf_{{2^n}}$. Moreover, in Lemma \ref{lem_deg_4}, $\omega_1=\frac{a^4+a^2+1}{(a^2+a+1)^2}=1$ and $\tr_{2^n}(\omega_1)=1$. Therefore, $G(x)$ is irreducible and Eq. (\ref{n-odd-12-eq-2}) has exactly one solution $x=0$ in $\gf_{{2^n}}\backslash\{1\}$. The proof has been finished.
\end{proof}

\subsection{The infinite classes for $n\equiv 0 \pmod 3$}

\begin{Th}
	\label{n=3m-1}
	Let $n=3m$ and $f(x)=x^{2^{2m}+2^m} + x^{2^{2m}+1} +x^{2^m+1} +x $. Then $f(x)$ is $2$-to-$1$ over $\gf_{{2^n}}$.
\end{Th}

\begin{proof}
	We suffice to show that for any $a\in\gf_{{2^n}}$, $f(x+a)+f(a)=0$, i.e., 
	\begin{equation}
	\label{n=3m-1-eq-1}
	\tr_m^{3m}\left( x^{2^m+1}+ax^{2^m}+a^{2^m}x \right) = x
	\end{equation}
	has exactly two solutions in $\gf_{{2^n}}$. From Eq. (\ref{n=3m-1-eq-1}), we have $x\in\gf_{2^m}$ and thus Eq. (\ref{n=3m-1-eq-1}) becomes $x^2+x=0$. Therefore, Eq. (\ref{n=3m-1-eq-1}) has exactly two solutions $x=0,1$.
\end{proof}

\begin{Th}
	\label{n=3m-2}
	Let $n=3m$ with $m\not\equiv 1\pmod 3$ and $f(x)= x^{2^{2m}+1}+x^{2^{m+1}}+x^{2^m+1}+x$. Then $f(x)$ is $2$-to-$1$ over $\gf_{{2^n}}$.
\end{Th}

\begin{proof}
	Let $y=x^{2^m}$ and $z=y^{2^m}$. Then $f=xz+y^2+xy+x$ and we suffice to show that for any $a\in\gf_{{2^n}}$, $f(x+a)+f(a)=0$, i.e., 
	\begin{equation}
	\label{n=3m-3-eq-1}
	xz+xy+y^2+(b+c+1)x+ay+az=0,
	\end{equation}
	where $b=a^{2^m}$ and $c=b^{2^{m}}$, has exactly two solutions in $\gf_{{2^n}}$. Raising Eq. (\ref{n=3m-3-eq-1}) into $2^m$-th power twice and we have
	\begin{equation}
	\label{n=3m-3-eq-2}
	xy+yz+z^2+(c+a+1)y+bz+bx=0
	\end{equation}
	and 
	\begin{equation}
	\label{n=3m-3-eq-3}
	yz+xz+x^2+(a+b+1)z+cx+cy=0.
	\end{equation}
		Computing $(\ref{n=3m-3-eq-1})+(\ref{n=3m-3-eq-2})+(\ref{n=3m-3-eq-3})$, we have $x+y+z=0$ or $1$. 
		
	(i) If $x+y+z=0$, plugging it into Eq. (\ref{n=3m-3-eq-1}) and Eq. (\ref{n=3m-3-eq-2}), we obtain $x^2+y^2+dx=0$ and $x^2+dy=0,$ where $d=a+b+c+1\in\gf_{{2^m}}$. If $d=0,$ then $x=0$, which is clearly one solution of Eq. (\ref{n=3m-3-eq-1}). Otherwise, we have $x^4+d^2x^2+d^3x=0$. Assume that $x=d\gamma,$ where  $\gamma^4+\gamma^2+\gamma=0$, then  $y=d\gamma^{2^m}$. On the other hand, plugging $x=d\gamma$ into $x^2+dy=0$ and we have $y=d\gamma^2$. Thus $\gamma^{2^m}=\gamma^2$. Since $\gcd\left(2^m-2,7\right)=1$ due to $m\not\equiv1\pmod3$, $\gamma =\gamma^{\gcd\left(2^m-2,7\right)} = 1$, which does not satisfy $\gamma^4+\gamma^2+\gamma=0$ clearly.  Hence, in the case, $x=0$ is the unique solution of Eq. (\ref{n=3m-3-eq-1}).
	
	(ii) If $x+y+z=1$, plugging it into Eq. (\ref{n=3m-3-eq-1}), we obtain $x^2+y^2+dx=a$, where $d=a+b+c\in\gf_{{2^m}}$. If $d=0$, then $x= (b+1)^{\frac{1}{2}}$, which is indeed the other solution of Eq. (\ref{n=3m-3-eq-1}). If $d\neq0$, we have 
	\begin{equation}
	\label{n=3m-3-eq-4}
	\left( \frac{y}{d}\right)^2+\left( \frac{x}{d} \right)^2+\frac{x}{d} =\frac{a}{d^2}.
	\end{equation}
 Let $L(x) = x^{2^{m+1}}+x^2+x \in\gf_{{2^n}}[x] $. Next we prove that $L$ is a permutation polynomial over $\gf_{{2^n}}$. Since $L$ is linearized, we suffice to prove that $L(x)= x^{2^{m+1}}+x^2+x =0$ has unique solution $x=0$ in $\gf_{{2^n}}$. Clearly, $\tr_m^n(L(x)) = \tr_{m}^n(x) = x+x^{2^m}+x^{2^{2m}} =0$ from $L(x)=0$. Thus $x^{2^{m+1}}+x^2=x^{2^{2m+1}} = x$, which means $x=0$ or $x^{2^{2m+1}-1}=1$. Since $\gcd\left( 2^{2m+1}-1, 2^n-1 \right) =1$ due to $m\not\equiv1\pmod 3$, we have $x=x^{\gcd\left( 2^{2m+1}-1, 2^n-1 \right)}=1$,  which does not satisfy $L(x)=0$ clearly. Hence, $L$ permutes $\gf_{{2^n}}$ and Eq. (\ref{n=3m-3-eq-4}) has exactly one solution in $\gf_{{2^n}}$. 
	
	All in all, Eq. (\ref{n=3m-3-eq-1}) has exactly  two solutions in $\gf_{{2^n}}$.
\end{proof}

In the last of this section, we summarize all twelve classes of $2$-to-$1$  quadrinomials in the following Table III. In addition, note that the experiment results by MAGMA show that our quadrinomials are not mutually QM-equivalent. 

\begin{table}[!htbp]
	\caption{$2$-to-$1$ quadrinomials $f(x)$ over $\gf_{{2^n}}$  }
	\centering
	\begin{tabular}{c cc c}	
		\toprule
		No. & Condition & $f(x)$ & Ref. \\
		\midrule
		$1$ & $n=2m+1$  & $x^{2^{m+1}+2}+x^{2^{m+1}}+x^2+x$ & Theorem \ref{n-odd-1} \\
		$2$ & $n=2m+1$ & $x^{2^{m+1}+2}+x^{2^{m+1}+1}+x^2+x $ & Theorem \ref{n-odd-2}\\
		$3$ & $n=2m+1$ & $x^{2^{m+2}+4}+x^{2^{m+1}+2}+x^2+x$ &Theorem \ref{n-odd-3}	\\
		$4$ & $n=2m+1$ & $x^{2^{n}-2^{m+1}+2}+x^{2^{m+1}}+x^2+x$ & Theorem \ref{n-odd-4}\\
		$5$ & $n=2m+1$ & $ x^{2^{n}-2}+x^{2^n-2^{m+1}}+x^{2^{n}-2^{m+1}-2}+x $& Theorem \ref{n-odd-5} \\
		$6$ & $n=2m+1$ & $x^{2^n-2}+x^{2^n-2^{m+1}}+x^{2^{m+1}-1}+x$& Theorem \ref{n-odd-6}\\
		$7$ & $n=2m+1$ & $x^{2^n-2}+x^{2^{n-1}+1}+x^{2^{n-1}-2}+x$& Theorem \ref{n-odd-7} \\
		$8$ & $n=2m+1$ & $x^{2^n-2}+x^{2^n-4}+x^3+x$& Theorem \ref{n-odd-8} \\
		$9$ & $n=2m+1$ & $x^6+x^4+x^3+x$ & Theorem \ref{n-odd-9}\\
		$10$ & $n=2m+1$ & $x^6+x^5+x^3+x$&  Theorem \ref{n-odd-10}\\
		$11$ & $n=3m$ & $x^{2^{2m}+2^m}+x^{2^{2m}+1}+x^{2^m+1}+x$& Theorem \ref{n=3m-1}\\
		$12$ & $n=3m$, $m{\not\equiv}1\pmod3$ & $x^{2^{2m}+1}+x^{2^{m+1}}+x^{2^m+1}+x$& Theorem \ref{n=3m-2} \\ 
		\bottomrule
	\end{tabular}
\end{table}

\section{Conclusions}
 A systematic study on $2$-to-$1$ mappings over finite fields was provided very recently by Mesnager and Qu \cite{MQ2019}.  But despite their numerous applications in many areas, very little is known in the literature about $2$-to-$1$ mappings up to now. Motivated by their work,  we have pushed further the study of $2$-to-$1$ mappings, particularly, over finite fields in even characteristic which are of particular interest. We contributed to the classification of $2$-to$1$ mappings and enriched the results about $2$-to-$1$ mappings using several elegant algebraic methods and tools. Firstly, we have completely determined  $2$-to-$1$ polynomials with degree $5$ over $\gf_{2^n}$ using the Hasse-Weil bound. In addition, using the multivariate method and the resultant of two polynomials,  we have presented two classes of $2$-to-$1$ trinomials and twelve classes of $2$-to-$1$  quadrinomials over $\gf_{2^n}$. 
%
%
%
%


\begin{thebibliography}{99}

\bibitem{Bartoli2018} D. Bartoli, \newblock On a conjecture about a class of permutation trinomials, \newblock {\em Finite Fields Appl.}, 52 (2018), 30-50.
 	\bibitem{SZM1967} E. R. Berlekamp, H. Rumsey and G. Solomon.
\newblock  On the solution of algebraic equations over finite fields. \newblock {\em Information And Control.}  10 (1967): 553-564.
 \bibitem{CM2011} C. Carlet and S. Mesnager, \newblock On Dillon's class $\mathcal{H}$ of bent functions, Niho bent functions and o-polynomials, \newblock {\em J. Comb. Theory, Ser. A}, 118 (2011), 2393-2410.  
 
 \bibitem{Dobbertin2002} H. Dobbertin, \newblock Uniformly representable permutation polynomials, \newblock {\em Sequence and their Applications-SETA 2001, Springer,} 2 (2002), 1-22. 
 
 	\bibitem{Hou2018} X. Hou, \newblock On a class of permutation trinomials in characteristic $2$, \newblock {\em Cryptogr. Commun.}, 2018,  https://doi.org/10.1007/s12095-018-0342-1.
 	 
 	 \bibitem{Houbook2018} X. Hou, \newblock Lectures on Finite Fields, \newblock {\em Grad. Stud. Math., vol. 190, Amer. Math. Soc., Providence, RI,} 2018.
 	 
 	 \bibitem{HKT2008} J. Hirschfeld, G. Korchm\'{a}ros and F. Torres, \newblock Algebraic curves over a finite field, \newblock {\em  Princeton University Press, STU - Student edition,} 2008. 
  \bibitem{LM1993} R. Lidl, G.L. Mullen and G. Turnwald, \newblock {Dickson Polynomials,} \newblock {\em Longman Scientific and Technical,} 1993.
 	\bibitem{LN1997} R. Lidl, H. Niederreiter, \newblock { Finite Fields}, 2nd ed. \newblock  {\em Cambridge Univ. Press, Cambridge}, 1997.
 
   	\bibitem{LQC2017} K. Li, L. Qu and X. Chen, \newblock New classes of permutation binomials and permutation trinomials over finite fields,\newblock {\em Finite Fields Appl.}, 43(2017), 69-85.

 
 	\bibitem{LQW2018} K. Li, L. Qu and Q. Wang, \newblock Compositional inverses of permutation polynomials of the form $x^rh\left(x^s\right)$ over finite fields, \newblock {\em Cryptogr. Commun.}, 11 (2019), 279-298.
 	\bibitem{LW1972} P. A. Leonard and K. S. Williams, \newblock Quartics over $\mathbb{GF}\left(2^n\right)$. \newblock {\em Proc. Am. Math. Soc.} 36 (1972): 347-350.

 	
 	\bibitem{MQ2019} S. Mesnager and L. Qu, \newblock { On two-to-one mappings over finite fields, }\newblock {\em IEEE Trans. Inf. Theory}, 2019,  https://doi.org/10.1109/TIT.2019.2933832.
  \bibitem{S1993} H. Stichtenoth, \newblock Algebraic function fields and codes, \newblock {\em Springer, Berlin}, 1993. 
    	\bibitem{WYD2017} D. Wu, P. Yuan, C. Ding and Y. Ma, \newblock Permutation trinomials over $\gf_{{2^m}}$,\newblock {\em Finite Fields Appl.}, 46(2017), 38-56.
  \end{thebibliography}
\end{document}